\documentclass[10pt, conference, letterpaper]{IEEEtran}
\IEEEoverridecommandlockouts
\usepackage[draft]{hyperref}
\usepackage[usenames,dvipsnames]{xcolor}

\usepackage{fontawesome,bm}
\usepackage{amsthm}
\usepackage{amssymb}
\usepackage{mathtools}
\usepackage{enumitem}

\usepackage{subcaption}

\newtheorem{proposition}{Proposition}
\newtheorem{theorem}{Theorem}
\newtheorem{corollary}{Corollary}[theorem]
\newtheorem{lemma}[theorem]{Lemma}
\newtheorem{definition}{Definition}
\theoremstyle{definition}
\newtheorem{remark}{Remark}

\usepackage{cite}

\usepackage{xspace}
\DeclareMathOperator{\interior}{int}

\newcommand{\EFRT}[1]{\gamma_{#1}^\star}

\newcommand{\ADVR}{\tilde{\mu}^\star}

\newcommand{\ifgraph}{G_I} 
\newcommand{\nstar} { \bar{n} }
\newcommand{\subhrm}[1]{ C_{#1}(t)}
\newcommand{\MPALGONAME}{\textsc{FAMIX-MS}\xspace} 
\newcommand{\ALGONAME}{\textsc{FAMIX-ND}\xspace} 
\newcommand{\LDF}{\textsc{LDF}\xspace} 
\newcommand{\GMS}{\textsc{GMS}\xspace} 
\newcommand{\DISTGREEDY}{\textsc{DistGreedy}\xspace} 
\DeclarePairedDelimiter{\ceil}{\lceil}{\rceil}

\newcommand{\MWS}{\textsc{MWS}\xspace}
\newcommand{\maxpacks}{a_{max}}
\newcommand{\NONEMPTY}{\mathcal{B}(t)}

\newcommand{\strat}{\mathcal{P}_{NC}(\mathcal{F})}
\newcommand{\bcdot}{\boldsymbol{\cdot}}

\newcommand{\be}{\begin{eqnarray}}
\newcommand{\ve}[1]{\mathbf{#1}}
\newcommand{\ee}{\end{eqnarray}}
\newcommand{\ben}{\begin{eqnarray*}}
\newcommand{\een}{\end{eqnarray*}}
\newcommand{\bfl}{\begin{flalign*}}
\newcommand{\efl}{\end{flalign*}}

\newcommand{\cpolicy}{\mathcal P_C}

\newcommand{\MAXDEAD}{d_{\max}} 
\newcommand{\USERSC}{\MATH{K}}
\newcommand{\USERS}{\MATH{\mathcal K}}
 
\newcommand{\STATEZ}{\mathcal{S}(t_0)} 
 
\newcommand{\STATET}{\mathcal{S}(t)} 
 
\newcommand{\TSTATE}[1]{\MATH{\bm{\tau}(#1)}}
\newcommand{\BSTATEE}{\Psi}
\newcommand{\BSTATE}[1]{
\ifx&#1&%
   \MATH{\Psi}
\else
   \MATH{\Psi(#1)}
\fi
}
\newcommand{\BSTATEP}[2]{\MATH{\Psi^{#2}(#1)}}
\newcommand{\TFSTATE}[1]{
\ifx&#1&%
   \MATH{\mathbf z}
\else
   \MATH{\mathbf z(#1)}
\fi
}
\newcommand{\TFSTATETUPLE}[1]{(\FSTATE{#1}, \TSTATE{#1})}
\newcommand{\FSTATE}[1]{\MATH{\mathbf{q}(#1)}}
\newcommand{\WSTATE}[1]{\MATH{\mathbf w(#1)}}
\newcommand{\WSTATEP}[2]{\MATH{\mathbf w^{#2}(#1)}}
\newcommand{\STATETUPLE}[1]{(\FSTATE{#1},\TSTATE{#1}, \BSTATE{#1}, \WSTATE{#1})} 
\newcommand{\STATETUPLEP}[2]{(\BSTATEP{#1}{#2}, \FSTATE{#1},\TSTATE{#1}, \WSTATEP{#1}{#2})}

\newcommand{\sumt}{{\sum_{t\in \mathcal F}}} 
\newcommand{\suml}{{\sum_{l\in \mathcal K}}} 
\newcommand{\sumM}[1]{{\sum_{l\in \M{#1}}}}

\DeclareMathOperator*{\argmax}{arg\,max}

\newcommand{\PA}{\textsc{p}}
\newcommand{\cb}[1]{\left\{#1\right\}}
\newcommand{\MIN}[2]{\min_{#1} \cb{#2}}
\newcommand{\MAX}[2]{\max_{#1} \cb{#2}}

\newcommand{\MATH}[1]{\ensuremath{#1}\xspace}
\newcommand{\br}[1]{\left[#1\right]}

\newcommand{\ETT}[1]{\mathbb E_{t,J}\br{#1}}
\newcommand{\EOT}[1]{\mathbb E_{t}\br{#1}}

\newcommand{\ES}{\mathbb E}
\newcommand{\FADRAN}{\textsc{F}}
\newcommand{\ALGRAN}{\textsc{R}}
\newcommand{\E}[1]{\mathbb E \br{#1}}
\newcommand{\ERF}[1]{\mathbb E^{\ALGRAN,\FADRAN} \br{#1}}
\newcommand{\ERFTJ}[1]{\mathbb E^{\ALGRAN,\FADRAN}_{t,J} \br{#1}}
\newcommand{\ERFT}[1]{\mathbb E^{\ALGRAN,\FADRAN}_{t} \br{#1}}
\newcommand{\EF}[1]{\mathbb E^{\FADRAN} \br{#1}}
\newcommand{\OPT}{\MATH{\mu^\star\xspace}}
\newcommand{\ALG}{\textsc{ALG}\xspace}

\newcommand{\GM}[1]{\mathcal G_{\M{#1}}}
\newcommand{\GADV}{\mathcal G_{\OPT}}
\newcommand{\GADVND}{\mathcal G_{\tilde \mu^\star}}
\newcommand{\GADVA}{\hat {\mathcal G}_{\OPT}}
\newcommand{\GADVAND}{\hat {\mathcal G}_{\tilde \mu^\star}}
\newcommand{\GADVE}[1]{\GADVA^{(#1)}(t)}
\newcommand{\GADVEND}[1]{\GADVAND^{(#1)}(t)}
\newcommand{\GADVEM}[1]{\GADVA^{(\M{#1})}(t)}
\newcommand{\GALG}{\mathcal G_{ALG}}
\newcommand{\PAT}{J}
\newcommand{\FRAME}{\MATH{\mathcal F}}
\newcommand{\frm}{\FRAME}
\newcommand{\FRAMEK}{\MATH{\mathcal F^{k}}}

\newcommand{\BUFFERT}[1]{\MATH{\mathcal B(#1)}}
\newcommand{\BUFFER}{\MATH{\BUFFERT{t}}}
\newcommand{\AMISE}{\MATH{\mathcal I}}
\newcommand{\AMIS}{\MATH{\mathcal M}}
\newcommand{\CAMIS}{\MATH{|\AMIS|}}
\newcommand{\CAMISE}{\MATH{|\AMISE|}}
\newcommand{\EPS}{\mathcal E}
\newcommand{\EPSC}{\mathcal E_0}
\newcommand{\EPSM}{\mathcal E_m}
\newcommand{\EPSG}{\mathcal E_F}
\newcommand{\W}[1]{W_{\M{#1}}(t)}
\newcommand{\WE}[1]{ \sumM{#1} w_l(t) \FP_l(t) }
\newcommand{\M}[1]{M^{#1}} 
\newcommand{\MT}[1]{\M{#1}(t)} 
\newcommand{\PI}{\pi}

\newcommand{\FP}{\MATH{q}}
\newcommand{\PHI}{\phi}
\newcommand{\DR}{\MATH{p}}
\newcommand{\DRC}{\MATH{\DR_l}}

\usepackage{dsfont,bm}

\newcommand{\GRATIO}{\frac{\CAMISE}{2 \CAMISE -1}}
\newcommand{\NID}[1]{\|\WSTATE{#1}\|}
\newcommand{\JFK}{\mathcal J(\mathcal{F}^{(k)})}
\newcommand{\JF}{\mathcal J(\mathcal{F})}

\newcommand{\I}[2][]{I_{#2}^{#1}}
\newcommand{\IS}{M}
\newcommand{\IF}{C}
\newcommand{\ACTIVELINK}[2]{\IF_{#1}(#2)=1}
\newcommand{\INACTIVELINK}[2]{\IF_{#1}(#2)=0}

\newcommand{\ACTIVES}[1]{\{l\in \USERS: \ACTIVELINK{l}{t}\}}
\newcommand{\CHR}{\MATH{\chi}}

\newcommand{\dref}[1]{(\ref{#1})}
\newcommand{\VAL}[3]{
V_{#1}(#3,#2,J)
}
\newcommand{\VALV}[3]{
\widetilde V_{#1}(#3,#2,J)
}

\newcommand{\NEW}[1]{\textcolor{black}{#1}}
\newcommand{\REMOVE}[1]{\textcolor{black}{#1}}

\def\BibTeX{{\rm B\kern-.05em{\sc i\kern-.025em b}\kern-.08em
    T\kern-.1667em\lower.7ex\hbox{E}\kern-.125emX}}

\begin{document}
\title{
Randomized Scheduling of Real-Time Traffic in Wireless Networks Over Fading Channels
}
\author{
	\IEEEauthorblockN{Christos Tsanikidis, Javad Ghaderi}
	\IEEEauthorblockA{Electrical Engineering Department, Columbia University}  
	\thanks{Emails: \{c.tsanikidis, jghaderi\}@columbia.edu.}
	}
\maketitle

\begin{abstract}
Despite the rich literature on scheduling algorithms for wireless networks, algorithms that can provide deadline guarantees on packet delivery for general traffic and interference models are very limited. In this paper, we study the problem of scheduling real-time traffic under a \textit{conflict-graph interference model} with \textit{unreliable links} due to channel fading. Packets that are not successfully delivered within their deadlines are of no value. We consider traffic (packet arrival and deadline) and fading (link reliability) processes that evolve as an \textit{unknown} finite-state Markov chain. The performance metric is efficiency ratio which is the fraction of packets of each link which are delivered within their deadlines compared to that under the optimal (unknown) policy. We first show a conversion result that shows classical non-real-time scheduling algorithms can be ported to the real-time setting and yield a constant efficiency ratio, in particular, Max-Weight Scheduling (\MWS) yields an efficiency ratio of $1/2$. We then propose randomized algorithms that achieve efficiency ratios strictly higher than $1/2$, by carefully randomizing over the maximal schedules. 
We further propose low-complexity and myopic distributed randomized algorithms, and characterize their efficiency ratio. Simulation results are presented that verify that randomized algorithms outperform classical algorithms such as \MWS and \GMS.     
\end{abstract}


\section{Introduction}
There has been vast research on scheduling algorithms in wireless networks which mostly focus on maximizing long-term throughput when packets have no strict delay constraints. Max-Weight Scheduling (MWS) policy has been shown to be throughput optimal in such settings, attaining any
desired throughput vector in the feasible throughput region~\cite{tassiulas1990stability}. Further, greedy scheduling policies such as LQF~\cite{joo2009understanding,dimakis2006sufficient}, or distributed policies such as CSMA~~\cite{ghaderi2010design, ni2012q, shah2010delay} have been proposed that alleviate the computational complexity of MWS and achieve a certain fraction of the throughout region.  
However, in many emerging applications, such as Internet of Things (IoT), vehicular networks, and edge computing, delays and deadline guarantees on packet delivery play an important role~\cite{lu2015real, song2008wirelesshart, gubbi2013internet}, as packets that are not received within specific deadlines are of little or no value, and are typically discarded. This \textit{discontinuity} in the packet value as a function of latency makes the problem significantly more challenging than traditional scheduling where packets do not have strict deadlines.

There is an increasing body of work attempting to address the above challenge, however they either assume a frame-based traffic model~\cite{houborkum09,houkum09,kumar10,srikant10,li2013optimal}, relax the interference graph constraints~\cite{singh2018throughput}, or use greedy scheduling approaches like LDF~\cite{kang2014performance,kang2015capacity}. In the frame-based traffic model, time is divided into frames, and packet arrivals and their deadlines during a frame are assumed to be \textit{known} at the beginning of the frame, and deadlines are constrained by the frame's length~\cite{houborkum09,houkum09,kumar10,srikant10,li2013optimal}. Under such assumptions, the optimal solution in each frame is a Max-Weight schedule. 
Note that unless the traffic is restricted to be synchronized across the users, such solutions are non-causal. The optimal scheduling policy (and the real-time throughout region) for general traffic patterns and interference graphs is unknown and very difficult to characterize.  
Largest-Deficit-First (LDF) is a causal policy which extends the well-studied Largest-Queue-First (LQF) from traditional scheduling to real-time scheduling. The performance of LDF has  been studied in terms of  \textit{efficiency  ratio},  which  is  the  fraction  of  the  real-time throughput region  guaranteed  by  LDF. Under \textit{i.i.d.}  packet arrivals and deadlines, with no fading, LDF was shown to achieve an efficiency ratio of at least $\frac 1 {\beta + 1}$~\cite{kang2014performance}, where $\beta$ is the \textit{interference degree} of the network (which is the maximum number of links that can be scheduled simultaneously out of a link and its neighboring links). 

Recently, the work~\cite{tsanikidis2020power} has shown that through randomization it is possible to design algorithms that can significantly improve the prior algorithms, in terms of both efficiency ratio and traffic assumptions. Specifically,~\cite{tsanikidis2020power} proposed two randomized scheduling policies: AMIX-ND for collocated networks with an efficiency ratio of at least $\frac{e-1}{e} \approx 0.63$, and AMIX-MS for general interference graphs with an efficiency ratio of at least $\GRATIO > \frac{1}{2}$ ($|\AMISE|$ is the number of maximal independent sets of the graph). However, the complexity of AMIX-MS can be prohibitive for implementation in large networks. Moreover, intrinsic wireless channel fading has not been considered in~\cite{tsanikidis2020power}, and packet transmission over a link is assumed to be always reliable. 

In this work, we consider an interference graph model of wireless network subject to fading, where packet transmissions over links are unreliable. We consider a joint traffic (packet arrival and deadline) and fading (link's success probability) process that evolves as an \textit{unknown} Markov chain over a finite state space. Can the existing traditional scheduling algorithms (which focus on long-term throughput with no deadline constraints) be used to provide guarantees for scheduling in this setting, \textit{without making frame-based traffic assumptions}? We show that interestingly the answer is yes, but fading and deadlines might significantly degrade their performance.  

\begin{figure} [t]
\centering
    \includegraphics[width=0.48\textwidth]{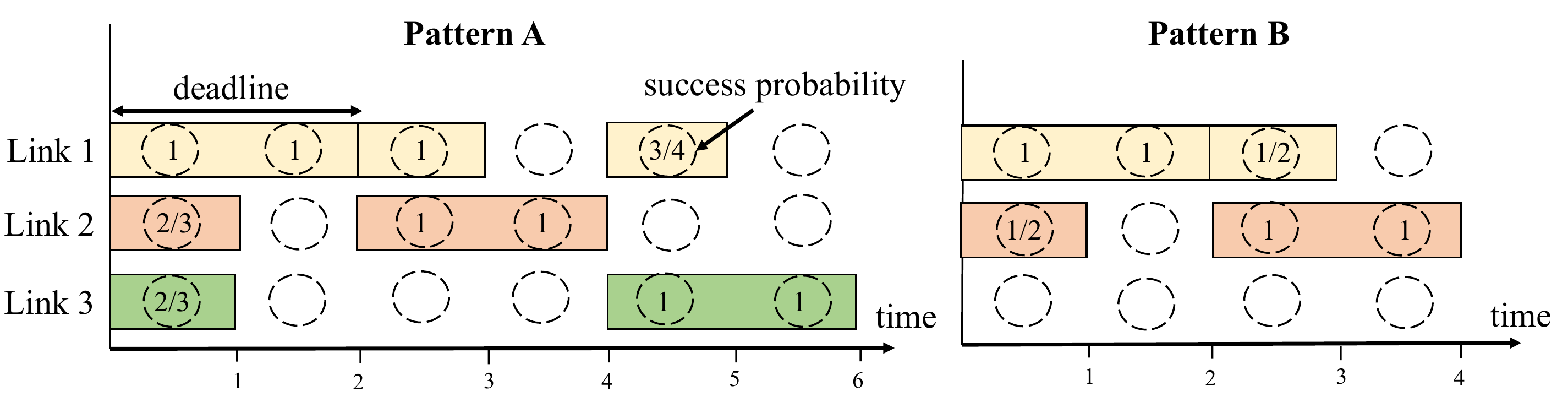}
\caption{An example of a Markovian traffic-and-fading process that alternates between two patterns. Each rectangle indicates a packet for a link. The left side of the rectangle corresponds to its arrival time, and its length corresponds to its deadline. The numbers in circles indicate the channel success probability of each link in each time slot.
}
\vspace{- 0.1 in}
\label{fig-traffic-fading}
\end{figure}
Introducing fading in deadline-constrained scheduling makes the problem very complicated. For example, consider the initial two time slots in Pattern B in traffic-fading process of Figure~\ref{fig-traffic-fading} with two interfering links: link $1$ has a packet with deadline $2$ and link $2$ has a packet with deadline $1$. Link 1's success probability is $1$ in the first two slots, and link 2's success probability is $0.5$ in the first time slot.  Not knowing the future traffic and fading, an opportunistic scheduler would prioritize link $1$ in the first time slot and subsequently in the second time slot, but the optimal policy would always schedule link $2$ in the first time slot and packet of link $1$ in the second time slot. A key insight of our work is that careful randomization in decision is crucial to hedge against the risk of poor decision due to lack of the knowledge of future traffic and fading.

\subsection{Contributions}
Our main contributions can be summarized as follows.
\begin{itemize}[leftmargin=*]
\item \textbf{Application of Traditional Scheduling Algorithms to Real-Time Scheduling.} Each non-empty link (i.e., with unexpired packets for transmission) is associated with a weight which is the product of its deficit counter and channel fading probability at that time, otherwise if the link is empty, the weight is zero. We show that any  algorithm that provides a $\psi$-approximation to Max-Weight Schedule (\MWS) under such weights, achieves
an efficiency ratio of at least $\frac{\psi}{\psi+1}$ for real-time scheduling under any Markov traffic-fading process.  As a consequence,  \MWS policy achieves an efficiency ratio of $\frac 1 2$, and \GMS (Greedy Maximal Scheduling) provides an efficiency ratio of at least $\frac 1 {\beta+1}$.
\item \textbf{Randomized Scheduling of Real-Time Traffic Over Fading Channels.}  We extend~\cite{tsanikidis2020power} to show the power of randomization for scheduling real-traffic traffic over fading channels in general and collocated networks. By carefully randomizing over the maximal schedules, the algorithms can achieve an efficiency ratio of \textit{at least} $\GRATIO >\frac {1}{2}$ in any general graph under any \textit{unknown}  Markov traffic-fading process. In the special case of a collocated network with i.i.d. channel success probability $q$, the algorithm can achieve an efficiency ratio of at least $( \frac{\FP}{1-e^{-\FP}} + 1-\FP)^{-1}$, which ranges from $0.5$ to $0.63$.

\item \textbf{Low-Complexity and Distributed Randomized Algorithms.}
To address the high complexity of the randomized algorithm in general graphs, we propose a low-complexity covering-based randomization and a myopic distributed randomization. Given a coloring of the graph using $\CHR$ colors, we can achieve an efficiency ratio of at least $\frac {1}{2 \CHR - 1}$, by randomizing over $\CHR$ schedules. 
Moreover, we show that a myopic distributed randomization, which is simple and easily implementable, can achieve an efficiency ratio of $\frac{1}{\Delta+1}$ in any graph with maximum degree $\Delta$. 
\end{itemize}
\subsection{Notations} Some of the basic notations used in this paper are as follows. Given $\mathbf y\in \mathbb R^n$,  $||\mathbf y|| = \sum_{i=1}^n |y_i|$. We use $\interior(A)$ to denote the interior of set $A$. $[x]^+=\max\{x,0\}$. $\mathds{1}(E)$ is the indicator function of event $E$. $|A|$ is the cardinality of set $A$. $\ES^{Y}[\cdot]$ is used to indicate that expectation is taken with respect to random variable $Y$. 
\section{Model and Definitions}\label{sec:model}
\textit{Wireless Network and Interference Model.}
Consider a set of \USERSC links, denoted by the set $\USERS$. 
We assume time is divided into slots, and in every slot $t$, each link $l \in \USERS$ can attempt to transmit at most one packet. To model interference between links, we use the standard \textit{interference graph} $G_I=(\USERS, E_I)$: Each vertex of $G_I$ is a link, and there is an edge $(l_1,l_2) \in E_I$ if links $l_1, l_2$ interfere with each other. Hence, no two links that transmit packets can share an edge in $G_I$.
Let $\AMISE$ be the set of all maximal independent sets of graph $G_I$. Also let $\NONEMPTY \subseteq \USERS$ denote the set of nonempty links (i.e., links that have packets available to transmit) at time $t$. We use $\IS(t)$ to denote the set of links scheduled at time $t$. By definition, $\IS(t)$ is a valid \textit{schedule} if links in $\IS(t)$ are nonempty and form an independent set of $G_I$, i.e.,
\be \label{eq:defschedule}
\IS(t) \subseteq (\NONEMPTY \cap D), \text{ for some }D \in \AMISE.
\ee
A schedule is said to be maximal if no nonempty link can be added to the schedule without violating the interference constraints. In this case, `$\subseteq$' in \dref {eq:defschedule} holds with `$=$'.

\textit{Fading Model.} Transmission over a link is unreliable due to wireless channel fading. To capture the channel fading over link $l$, we use an ON-OFF model where link $l$ at time $t$ is ON ($\IF_{l}(t)=1$) with probability $q_l(t)$, otherwise it is OFF ($\IF_{l}(t)=0$). 
If scheduled, transmission over link $l$ at time $t$ is successful only if link $l$ is ON. 
Let $\FSTATE{t}=( q_l(t), l\in \USERS)$ be the vector of success probabilities of the links.
We assume that at any time slot $t$, the link's success probability is known to the scheduler before making a decision. At the end of time slot $t$, link's transmitter receives a feedback from its receiver indicating whether transmission was successful or not.  A special case of this model is when  $q_l(t)\in \{0,1\}$, $\forall l\in \USERS$, in which case the channel state $C_l(t)$ is deterministically known to the scheduler, which requires periodic channel state estimation. Another special case is when $C_l(t)$ is i.i.d. with some probability $q_l$ which eliminates the need for periodic estimation.  Similar models have been used in \cite{hou2013scheduling,singh2018throughput}. 

Overall, we define $I(t)=(\I{l}(t), l\in \USERS)$ to denote the successful packet transmissions over the links at time $t$. Note that by definition, $\I{l}(t)= \mathds{1}(l \in \IS(t)) \IF_{l}(t)$, where $\IS(t)$ is the valid schedule at time $t$, and  $\IF_{l}(t)$ is the channel state of link $l$.

\textit{Traffic Model.}
We assume a single-hop real-time traffic. We use $a_l(t) \leq a_{\text{max}}$ to denote the number of packet arrivals at link $l$ at time $t$, where $a_{\text{max}} < \infty$ is a constant. Each arriving packet has a deadline which indicates the maximum delay that the packet can tolerate before successful transmission.  A
packet with deadline $d$ at time $t$ has to be successfully transmitted
before the end of time slot $t + d - 1$, otherwise it will be discarded. We define the traffic process $\TSTATE{t}=(\tau_{l,d}(t),  d=1,\cdots,\MAXDEAD,l \in \USERS)$, where $\tau_{l,d}(t)$ is the number of packets with deadline $d$ arriving to link $l$ at time $t$, and $\MAXDEAD < \infty$.

\textit{Traffic and Fading Process.} 
In general, we assume that the joint traffic and fading process $\TFSTATE{t} = \TFSTATETUPLE{t}$ evolves as an ``\textit{unknown}'' irreducible Markov chain over a finite state space $\mathcal{Z}$. See Figure~\ref{fig-traffic-fading} for an example of a Markovian traffic-fading process.

Without loss of generality, we make the following assumption to make this Markov chain \textit{non-trivial}: For every link $l$, there are two states ${\TFSTATE{}}^l, {\TFSTATE{}'}^l \in \mathcal{Z}$  such that ${\TFSTATE{}}^l$ has a packet arrival with deadline $d$, and $ {\TFSTATE{}'}^l$ has $q_{l}>0$, and there is a positive probability that $\TFSTATE{t}$ can go from ${\TFSTATE{}}^l$ to ${\TFSTATE{}'}^l$ in at most $d$ time slots. This assumption simply states that it is possible to successfully transmit some packets of every link $l$ within their deadlines. If a link does not satisfy this condition, we can simply remove it from the system. Note that $\TFSTATE{t}$ is an irreducible finite-state  Markov chain, hence it is positive recurrent \cite{dynkin2012theory}, and time-average of any bounded function of $\TFSTATE{t}$ is well defined, in particular the packet arrival rate for link $l$:
\be \label{eq:arrival rate}
&\lim_{t\to \infty} \frac{1}{t}{\sum_{s=1}^{t} a_l(s)}=:\overline{a}_l.
\ee

\textit{Buffer Dynamics.} The buffer of link $l$ at time $t$, denoted by $\Psi_l(t)$, contains the existing packets at link $l$ which have not expired yet and also the newly arrived packets at time $t$. The remaining deadline of each packet in $\Psi_l(t)$ decreases by one at every time slot, until the packet is \textit{successfully} transmitted or reaches the deadline $0$, which in either case the packet is removed from $\Psi_l(t)$.
We also define $\Psi(t)=(\Psi_l(t); l \in \USERS)$.  

\textit{Delivery Requirement and Deficit.} As in \cite{houborkum09,houkum09,kumar10,srikant10,kang2014performance}, we assume that there is a minimum delivery ratio requirement \DRC (QoS requirement) for each link $l\in \USERS$. This means we must successfully deliver at least \DRC fraction of the incoming packets on each link $l$ \textit{within their deadlines}. Formally,
\be \label{eq:delivery ratio}
&\liminf_{t\to \infty} \frac{\sum_{s=1}^tI_l(s)}{\sum_{s=1}^ta_l(s)}\geq \DRC.
\ee

We define a deficit $w_l(t)$ which measures the number of successful packet transmissions owed to link $l$ up to time $t$ to fulfill its minimum delivery ratio. As in \cite{kang2014performance,srikant10, tsanikidis2020power}, the deficit evolves as
\be \label{eq:deficit queue}
& w_l(t+1)=\Big[w_l(t)+\widetilde{a}_l(t)-I_l(t)\Big]^+,
\ee
where $\widetilde{a}_l(t)$ indicates the amount of deficit increase due to packet arrivals. For each packet arrival, we should increase the deficit by \DRC on average. For example, we can increase the deficit by exactly \DRC for each packet arrival to link $l$, or use a coin tossing process as in~\cite{kang2014performance,srikant10}, i.e., each packet arrival at link $l$ increases the deficit by one with the probability \DRC, and zero otherwise. We refer to $\widetilde{a}_l(t)$ as the \textit{deficit arrival} process for link $l$. Note that it holds that
\be \label{eq:deficitarr}
& \lim_{t\to \infty} \frac{1}{t}\sum_{s=1}^{t} {\widetilde{a}_l(s)}=\overline{a}_l \DRC:=\lambda_l, \ l\in \USERS.
\ee
We refer to $\lambda_l$ as the deficit arrival rate of link $l$. Note that an arriving packet is always added to the link's buffer, regardless of whether and how much deficit is added for that packet. Also note that in \dref{eq:deficit queue} each time a packet is transmitted successfully from link $l$, i.e., $I_l(t)=1$, the deficit is reduced by one. The dynamics in~\dref{eq:deficit queue} define a deficit queueing system, with bounded increments/decrements, whose stability, e.g., in the sense
\be \label{eq:strong stability}
& \limsup_{t\to \infty} \frac{1}{t}\sum_{s=1}^t\mathbb{E}[w_l(s)] <\infty,
\ee
implies \dref{eq:delivery ratio} holds~\cite{neely2010queue}. Define the vector of deficits as $\WSTATE{t}=(w_l(t), l \in \USERS)$. The system state at time $t$ is then defined as
\begin{equation}\label{eq:statedef}
\STATET=\STATETUPLE{t}.
\end{equation}
%

\textit{Objective.} Define $\cpolicy$ to be the set of all causal policies, i.e., policies that do not know the information of future arrivals, deadlines, and channel success probabilities in order to make scheduling decisions.
For a given traffic-fading process $\TFSTATE{t}$, with fixed $\overline a_l$, defined in \dref{eq:arrival rate}, we are interested in causal policies that can stabilize the deficit queues for the largest set of delivery rate vectors $\ve{p}=(\DRC, l \in \USERS)$, or equivalently largest set of $\bm{\lambda}=(\lambda_l := \overline a_l \DRC, l \in \USERS)$ possible. 
For a given traffic process, we say the rate vector $\bm{\lambda}=(\lambda_l,l\in \mathcal K)$ is supportable under some policy $\mu \in \cpolicy$ if all the deficit queues remain stable for that policy. Then one can define the supportable real-time rate region of the policy $\mu$ as 
\be
& \Lambda_{\mu}=\{\bm{\lambda} \geq 0: \bm{\lambda} \text{ is supportable by } \mu\}. 
\ee
The supportable \textit{real-time rate region} under all the causal policies is defined as $\Lambda = \bigcup_{\mu \in \cpolicy} \Lambda_{\mu}$. The overall performance of a policy $\mu$ is evaluated by the \textit{efficiency ratio} defined as 
\be \label{eq:capacity region}
\gamma_\mu^\star = \sup \{\gamma: \gamma \Lambda \subseteq \Lambda_\mu\}.
\ee
For a casual policy $\mu$, we aim to provide a \textit{universal lower bound} on the efficiency ratio that holds for ``all'' Markovian traffic-fading processes, \textit{without} knowing the Markov chain.

For clarity, we will use $\EF{\cdot}$ to denote expectation with respect to the \textit{random outcomes of the fading channel}, and $\ES^\ALGRAN[{\cdot}]$ to denote expectation with respect to the \textit{random decisions of a randomized algorithm} \ALG, whenever applicable.


\section{Scheduling Algorithms and Main Results}
Recall that $\AMISE$ is the set of all maximal independent sets of interference graph $G_I$, and $\BUFFER$ is the set of links that have packets to transmit at time $t$. At any time $t$, we define the set of maximal schedules $\AMIS(t)=\{D \cap \BUFFER, D \in \AMISE\}$.

Define the \textit{gain} of maximal schedule $M \in \AMIS(t)$ to be the total deficit of packets transmitted successfully at time $t$, i.e., $\GM{}(t) := \sumM{} \IF_l(t) w_l(t)$. We define the \textit{weight} of maximal schedule $M$ to be its expected gain, conditional on $\STATET$, i.e.,
\be
W_{M}(t) := \EF{\GM{}(t)|\STATET} = \sumM{} w_l(t) q_l(t) \label{eq:weight-defn}.
\ee
We use  $M^{\star}(t) := \arg\max_{M\in \AMIS(t)}\W{}$ to denote the \textit{Max-Weight Schedule} (\MWS) at time $t$, and $W^{\star}(t) := W_{M^{\star}(t)}(t)$ to denote its weight. 

Given a Markov policy \ALG, let $\PI_M(t)$ be the probability that \ALG selects maximal schedule $M$ at time $t$. Hence, the probability that link $l$ is scheduled at time $t$ is
\be \label{def:phi}
&\PHI_l(t)=\sum_{M \in \AMIS(t)} \PI_M(t) \mathds{1}( l \in M).
\ee
Further, the expected gain of \ALG is given by
\be
   \ERF{\GALG(t)|\STATET} &=&\sum_{\M{}\in \AMIS(t)} \PI_{\M{}}(t) \W{} \label{eq:glag-generic2}\\
    &=&\ \ \sum_{l\in \USERS} \PHI_l(t) \FP_l(t) w_l(t).
    \label{eq:galg-generic}
\ee

Without loss of generality, we consider natural policies that transmit the earliest-deadline packet from every selected link in the schedule. This is because, similar to \cite{tsanikidis2020power}, if a policy transmits a packet that is not the earliest-deadline, that packet can be replaced with the earliest-deadline packet of that link, and this only improves the state. Similarly, the optimal policy always selects a maximal schedule at any time. 

\subsection{Converting Classical Non-Real-Time Algorithms for Real-Time Scheduling} 
The following theorem allows us to convert non-real-time scheduling policies to real-time scheduling policies, hence enabling the use of numerous policies from the literature of  traditional non-real-time scheduling. \NEW{More specifically, policies whose expected gain at any time $t$ is $\psi$-fraction of the gain of \MWS (i.e. $W^\star(t)$) yield an efficiency ratio of at least $\psi/(\psi+1)$. The result is stated formally in Theorem~\ref{thm:convert}}.

\begin{theorem}
\label{thm:convert}
Consider any policy $\ALG$ such that, at every time $t$, 
$\ERF{\GALG(t)|\STATET} \geq  \psi W_{M^\star}(t)$, whenever $\|\WSTATE{t}\| \geq W'$, for some finite $W'$. Then 
\[\EFRT{\mu}\geq \frac{\psi}{\psi+1}. \]
\end{theorem}
Note that this conversion results in deadline-oblivious policies which can be preferred in cases where information about the deadlines of packets is either not accurate or not available. \NEW{We remark that for certain policies the result of Theorem~\ref{thm:convert} is tight as seen in Corollary~\ref{cor:mws} below, whereas for other policies the bound can be considerably loose.} 

\begin{corollary} \label{cor:mws}
\MWS policy provides efficiency ratio $\EFRT{\MWS}= \frac 1 2$ for real-time scheduling under Markov traffic-fading processes.
\end{corollary}
\begin{proof}
Using Theorem~\ref{thm:convert} for \MWS with $\psi=1$, we directly obtain $\EFRT{\MWS} \geq 1/2$. We can get the the opposite inequality through an adversarial example. If we consider a simple network with two interfering links without fading, then \MWS reduces to \LDF, which has been shown to have $\EFRT{}\leq 1/2$ for Markovian traffic with deterministic deficit admission~\cite{kang2014performance,tsanikidis2020power} (recall that efficiency ratio $\EFRT{}$ is defined as a universal bound for all traffic-fading processes for any given  graph).
\end{proof}
Consider a Greedy Maximal Scheduling (\GMS) policy defined as follows: Order the nonempty links $\BUFFER$ in the decreasing order of the product $w_l(t) \FP_l(t)$. Then construct a schedule recursively by including the nonempty link with the largest $w_l(t) \FP_l(t)$, removing the interfering links, and repeating the same procedure over the remaining links.

The following corollary extends the result of~\cite{kang2014performance} which was shown for i.i.d. traffic without fading for \LDF. 
\begin{corollary} \GMS provides an efficiency ratio $\EFRT{\GMS}\geq \frac 1 {\beta+1}$ for real-time scheduling under any Markov traffic-fading process, where $\beta$ is  the  interference  degree  of  the graph. 
\end{corollary}
\begin{proof}
Suppose $M=\{l_1,...,l_n\}$ is the scheduled constructed by \GMS, where $l_i$ is the link selected by \GMS in iteration $i$. It is sufficient to show that $\sum_{l \in M} w_l(t) \FP_l \geq \frac 1 \beta \W{\star}$, as in this case we can apply Theorem~\ref{thm:convert} with $\psi=1/\beta$ and consequently get $\EFRT{}\geq \frac{1}{\beta+1}$.
Consider the first iteration where \GMS selects $l_1$. The Max-Weight schedule $M^\star$ will either include link $l_1$, or at most $\beta$ links from the interfering links of $l_1$, but since $l_1$ has the highest $w_l(t) \FP_l$, it holds that $w_{l_1}(t) \FP_{l_1} \geq \frac 1 \beta \sum_{l\in N(l_1)\cap M^\star} w_l(t) \FP_{l}$, where $N(l)$ denotes the neighborhood of link $l$, including link $l$ itself, in graph $G_I$. Then we remove all the links in $N(l_1)$ and repeat the same argument for $\USERS \setminus N(l_1)$ and obtain
$w_{l_2}(t) \FP_{l_2} \geq \frac 1 \beta \sum_{l\in (N(l_2)\setminus N(l_1))\cap M^\star} w_l(t) \FP_{l}$. Repeating this argument, and summing the obtained inequalities over $l_1,\cdots,l_n$, we obtain the required inequality.
\end{proof}
Note that many other results from traditional scheduling literature can be converted using Theorem~\ref{thm:convert}. For example, the \DISTGREEDY policy \cite{joo2015distributed} would also obtain $\EFRT{\DISTGREEDY} \geq \frac{1}{\beta+1}$ but has a lower complexity than \GMS. 

\subsection{Randomized Scheduling Algorithms}\label{sec:genrand}
In this section, we extend {AMIX} policies for collocated networks and general networks introduced in~\cite{tsanikidis2020power} to incorporate fading. We refer to the generalized policies as FAMIX (\textit{Fading-based Adaptive MIX}). \NEW{We describe these policies below.}
\subsubsection{FAMIX-MS: Randomized Scheduling in General Graphs} 

Let $R := |\AMIS(t)|$ be the number of maximal schedules at time $t$. We index and order the maximal schedules such that $M^{(i)}\in \AMIS(t)$ has the $i$-th largest weight (based on \dref{eq:weight-defn}), i.e., 
\be
& \W{(1)}\geq \W{(2)}\cdots \geq \W{(R)}.
\ee
Define the \textit{subharmonic} average of weight of the first $n$ maximal schedules, $n \leq R$,  to be
\be \label{def:subhrm}
\subhrm{n} = \frac {n-1}{\sum_{i=1}^{n} \left(\W{(i)}\right)^{-1}}. 
\ee
Then select schedule  $M^{(i)}$ with probability
\be \label{eq:mwisprob}
    \PI^{\nstar}_{M^{(i)}}(t)\equiv \PI^{\nstar}_{i}(t) = \begin{dcases}
         1- \frac {\subhrm{\nstar}} {\W{(i)}}, &  1\leq i\leq \nstar \\
        0, & \nstar<i\leq R 
    \end{dcases}
\ee
where $\nstar$ is the largest $n$ such that $\{\PI^{n}_i(t), 1 \leq i \leq R\}$ defines a valid probability distribution, i.e., $\PI^{n}_i(t) \geq0$, and $\sum_{i=1}^R\PI^{n}_i(t)=1$. 

\begin{remark}It has been shown in \cite{tsanikidis2020power} that $\nstar$ in such a distribution can be found through a  binary search. Essentially the difference between FAMIX-MS and AMIX-MS in~\cite{tsanikidis2020power} lies in different definitions for the weight of a schedule. 
\end{remark}
\begin{theorem}\label{main-theorem-general}
In a general interference graph $\ifgraph$ with maximal independent sets $\AMISE$, the efficiency ratio of \MPALGONAME\ is 
\[\gamma_{\MPALGONAME}^\star \geq \GRATIO >\frac{1}{2}.\]
\end{theorem}
\begin{remark}
Theorem~\ref{main-theorem-general} shows that with randomization we can do \textit{strictly} better than \MWS (Corollary~\ref{cor:mws}). \NEW{In particular, in the case of a complete bipartite inteirference graph \MPALGONAME\ (where there are two maximal independent sets) yields an efficiency ratio of at least $2/3$ while \MWS yields $1/2$.}
\end{remark}
\subsubsection{FAMIX-ND: Randomized Scheduling in Collocated Graphs}

Extending AMIX-ND~\cite{tsanikidis2020power} to fading channels is more challenging. In particular, the derivation in \cite{tsanikidis2020power} relied on two main ideas: (1) it is sufficient to consider only a restricted set of ``non-dominated'' links for transmission, and (2) there is an ordering among the non-dominated links such that given two non-dominated links, having packet in buffer from one of the links is always preferred to that from the other link. Finding such domination relationship for a general Markov fading process is difficult.  Here, we describe an extension under a simplified fading process, where  $\FP_l(t)=\FP_l=\FP$, i.e., the links' success probabilities are fixed and equal.
We allow the channel state across links to be either independent or positively correlated,
i.e., $\Pr[\IF_{l_2}(t)=1|\IF_{l_1}(t)=1]\geq \FP_{l_2}$. 
The above setting could be a reasonable approximation in collocated networks where links have similar reliabilities, and an active channel for one link implies a better condition for the  overall shared wireless medium.

Let $e_l(t)$ denote the deadline of the earliest-deadline packet of link $l$ at time $t$. We say that link $l_1$ dominates link $l_2$ at time $t$ if $e_{l_1}(t)\leq e_{l_2}(t)$, $w_{l_1}(t) \geq w_{l_2}(t)$, i.e., link $l_1$ is more urgent and has a higher deficit. Based on this definition, the set of non-dominated links at any time can be found though a simple recursive procedure as in~\cite{tsanikidis2020power}\NEW{, i.e., add the largest-deficit nonempty link, remove all the links dominated by it from consideration and repeat the process for the remaining links}. The following theorem describes \ALGONAME and its efficiency ratio for a collocated network with a channel success probability $q$. 

\begin{theorem} \label{main-theorem-collocated}
Consider a collocated network, where $q_l=q$, $\forall l \in \USERS$,  and channels of links are independent or positively correlated.  Order and re-index the non-dominated links such that
\[w_1(t) \geq w_2(t) \geq ... \geq w_n(t).\]
Starting from $i=1$, assign probability $\pi_i(t)$ to the $i$-th non-dominated link, 
\be \label{eq:ndprob}
\pi_i(t) = \min \Big\{\frac 1 \FP \Big(1 - \frac {w_{i+1}(t)}{w_i(t)}\Big), 1 - \sum_{j < i } {\pi_j(t)} \Big\}.
\ee
\ALGONAME selects the $i$-th non-dominated link with probability $\pi_i(t)$ and transmits its earliest-deadline packet. Then
\be \label{eq:famix-nd-ratio}
\gamma_{\ALGONAME}^\star \geq 
\Big({ \frac{\FP}{1-e^{-\FP}} + (1-\FP)}\Big)^{-1}:=h(\FP) .
\ee
\end{theorem}
\begin{remark}
Note that due to channel uncertainty, \ALGONAME boosts the probability of larger-weight links. Intuitively, as $q$ becomes small, deadlines of packets are ``effectively'' reduced, as each packet will need to be transmitted several times before success. 
\end{remark}
\begin{remark} 
The  lower bound $h(\FP)$ on efficiency ratio in \dref{eq:famix-nd-ratio} is a monotone function of $\FP$, which increases from $0.5$ and to $\frac{e-1}{e}$, as $\FP$ goes from $0$ to $1$. For $\FP=1$, this recovers the result of \cite{tsanikidis2020power} for non-fading channels, i.e., $\gamma_{\text{AMIX-ND}}^\star \geq \frac{e-1}{e}$. 

In the case of unequal $\FP_i \in (\FP_{\min}, \FP_{\max})$, 
using \dref{eq:ndprob} by replacing $\FP$ with $\FP_{\min}$,
will give an efficiency ratio of at least $
({ \frac{\FP_{\max}}{1-e^{-\FP_{\min}}} + 1-\FP_{\max}})^{-1}:=h(\FP_{\min},\FP_{\max}). 
$ 
Depending on $\FP_{\max}, \FP_{\min}, \USERSC$, we can choose either  \ALGONAME or \MPALGONAME and achieve $\gamma^\star \geq \max \{h(\FP_{\min},\FP_{\max}), \frac{\USERSC}{2\USERSC-1}\}$. 
\end{remark}

\subsection{Low-Complexity and Distributed Randomized Variants}
The general algorithm \MPALGONAME in Section~\ref{sec:genrand} potentially randomizes over all the maximal schedules. This can be computationally expensive in large networks that may have many maximal schedules. In this section, we design variants that only need to consider a subset of the maximal schedules, or are distributed.
\subsubsection{Covering-Based Randomized Algorithms} \label{section-efficient}

We propose two variants that only need to consider a subset of the maximal schedules.
Proposition \ref{coloring-proposition} below states a sufficient condition under which randomization over a subset of the maximal schedules can provide a related approximation on the efficiency ratio. 

\begin{proposition} \label{coloring-proposition}
Consider policy $\ALG=\MPALGONAME|_{\AMIS_0(t)}$ which, at any time $t$, selects a schedule from a subset $\AMIS_0(t) \subseteq \AMIS(t)$, according to probabilities of \MPALGONAME computed for $\AMIS_0(t)$. Suppose that for every $M \in \AMIS(t)\setminus \AMIS_0(t)$,
\be
\psi(t)  \WE{} \bar \PHI_l(t)  \leq  \max_{ \M{\prime} \in \AMIS_0(t) }  \sum_{l \in \M{\prime}}  \FP_l(t) w_l(t) \bar \PHI_l(t) 
  \nonumber\\
+  (1-\psi(t)) \ERF{\GALG(t)|\STATET}, \label{eq:colortricksuffnew}
\ee
for some $\psi(t) \in (0,1]$, where $\PHI_l(t)$ was defined in \dref{def:phi} and $\overline \PHI_l(t) = 1-\PHI_l(t)$. Then
\be
& \EFRT{\ALG} \geq \MIN{t\geq0}{\psi(t) \frac {|\AMIS_0(t)|}{2|\AMIS_0(t)|-1} }.
\ee
\end{proposition}
\NEW{This condition in Proposition~\ref{coloring-proposition} has parameter $\psi(t)$. If the condition is satisfied for higher $\psi(t)$ the resulting bound is improved.}

\NEW{Next, by focusing on a special case in which Condition~\dref{eq:colortricksuffnew} holds, we can design provably efficient policies by considering a small set of schedules such that any other maximal schedule can be covered by them.}
\begin{lemma}\label{lemma:suffcond}
Suppose every $\M{} \in \AMIS(t)\setminus \mathcal \AMIS_0(t)$ is \textit{covered} by at most $\zeta$ maximal schedules from $\AMIS_0(t)$, i.e., 
\be \label{eq:coverdef}
M \subseteq \cup_{M^{\prime} \in S_M} \M{\prime}, \text{ for some } S_M \subseteq \AMIS_0(t):\ |S_M|\leq \zeta. 
\ee
Then Condition \dref{eq:colortricksuffnew} holds with fixed $\psi(t)=\frac{1}{\zeta}$.
\end{lemma}
\begin{proof}
Using the covering definition~\dref{eq:coverdef}, we have
\begin{flalign} 
  \sum_{l \in \M{}} \FP_l(t) w_l(t) \bar \phi_l(t) 
  &\leq 
 \sum_{\M{\prime}\in S_M} {\sum_{l \in \M{\prime}}  \FP_l(t) w_l(t) \bar \phi_l(t)}  \nonumber \\
 &\leq   \zeta \max_{\M{\prime}\in \AMIS_0(t)} {\sum_{l \in \M{\prime}}  \FP_l(t) w_l(t) \bar \phi_l(t)},
\end{flalign}
and hence condition \dref{eq:colortricksuffnew} trivially holds for $\psi=\frac{1}{\zeta}$.
\end{proof}
In general, $\AMIS_0(t)$ can be adaptive and constructed based on the link deficits and fading probabilities. Here, we apply Proposition~\ref{coloring-proposition} and Lemma~\ref{lemma:suffcond} for a constant $\psi(t)=\psi$ and a family $\AMIS_0(t)$ induced by fixed subset of the independent sets $\AMISE_0 \subseteq \AMISE$, i.e., $\AMIS_0(t) = \{D \cap \BUFFER, D \in \AMISE_0\}$. With minor abuse of notation, we refer to such algorithms as $\MPALGONAME|_{\AMISE_0}$. 
Below, we present two such covering-based algorithms.
\begin{corollary}[{Coloring-based Randomization}]\label{cor:coloring}
Consider a coloring of graph $G_I$ with $\CHR$ colors, which partitions the vertices of $G_I$ into $\CHR$  independent sets $\{D_1^\prime,...,D_{\CHR}^\prime\}$.
Extend these independent sets arbitrarily so they are maximal $\{D_1,...,D_{\CHR}\}:=\AMISE_0 $. Then $\EFRT{\MPALGONAME|_{\AMISE_0}}\geq \frac{1}{2 \CHR-1}$.
\end{corollary}
\begin{proof}
Any maximal schedule in $\AMIS\setminus \AMIS_0$ can be covered by at most $\CHR$ maximal schedules in $\AMIS_0$, thus $\psi= \frac{1}{\CHR}$ by Lemma~\ref{lemma:suffcond}. Further, $ \frac{|\AMIS_0(t)|}{2|\AMIS_0(t)|-1}  \geq \frac{\CHR}{2 \CHR-1}$. Thus applying Proposition~\ref{coloring-proposition}, we get the stated efficiency ratio. 
\end{proof}
\begin{remark}In general, finding an efficient coloring might be computationally demanding, but it needs to be done only once for a given $G_I$. There are many interesting families of graphs for which coloring can be solved efficiently. For example, for a (not necessarily complete) bipartite graph (e.g. a tree) where $\CHR=2$, we obtain $\EFRT{} \geq \frac 1 3$. This performs much better than \LDF whose efficiency ratio in a tree with maximum degree $\beta$ is $\EFRT{} \geq \frac 1 {2 \sqrt{\beta-1}+1}$ in the case without fading (which is a special case of our setting).
Another family that admits an efficient coloring are planar graphs where, by the four-color theorem, always have a 4-coloring which can be found in polynomial time \cite{robertson1996new}. Further, we remark that the independent sets $D_i^\prime$ could be extended adaptively at every time $t$, e.g., using \GMS.
\end{remark}
\REMOVE{
\begin{corollary}[Partition-based Randomization]\label{cor:partition}
Partition the vertices of graph $G_I$ into two sets $\mathcal K_1$ and $\mathcal K_2$, each with at most $\ceil{K/2}$ vertices. Consider the set of maximal independent sets of each induced subgraph, denoted by $\AMISE_1', \AMISE_2'$. Then extend these independent sets arbitrarily so that they form maximal independent sets $\AMISE_1, \AMISE_2$ for the whole graph $G_I$.  Then $\EFRT{\MPALGONAME|_{\AMISE_1 \cup \AMISE_2}}\geq \frac 1 2 \frac {|\AMISE_1|+|\AMISE_2|}{2(|\AMISE_1|+|\AMISE_2|)-1}>  \frac 1 4$.  
\end{corollary}
\begin{proof}
By construction, every maximal independent set in $G_I$ can be covered by at most 2 maximal independent sets in ${\AMISE_1 \cup \AMISE_2}$, and hence we have $\psi = 1/2$. Following similar arguments as in Corollary~\ref{cor:coloring}, we obtain the result.
\end{proof}
\begin{remark}By Corollary~\ref{cor:partition}, we only need to randomize over at most $|\AMISE_1|+|\AMISE_2|$ maximal schedules at any time  which can be much smaller than  $|\AMISE|$ in large graphs.  
The corollary can be generalized to more than $2$ partitions to allow a trade-off between the guaranteed efficiency ratio and complexity. 
\end{remark}
}


\subsubsection{Myopic Distributed Randomized Algorithm}
We present a simple distributed algorithm that has constant complexity.  

Assume each slot is divided in two parts, a control part of duration $T_C$, and a packet transmission part with duration normalized to $1$.
At the beginning of the control phase, every non-empty link $l \in \BUFFER$ starts a timer $T_l \sim Exp(\nu_l)$, where $Exp(\nu)$ denotes an exponential distribution with rate $\nu$.  
Once the timer of a link $l$ runs down to zero, it broadcasts an announcement informing its neighbors that it will participate in data transmission, \textit{unless} it has heard an earlier announcement from its neighboring links, or the control phase ends. 

Given any $\delta \in (0,1)$, let $T_C= \max_{\nu_l}\{\frac{-\log(\delta)}{\nu_l}\}$. The next corollary states the efficiency ratio for the uniform timer rates. 
\begin{corollary}
\label{myopic-generalresult}
Consider the myopic randomized algorithm where every link $l$ has the same timer rate $\nu_l = \nu$. If the maximum degree of $G_I$ is $\Delta$, then 
\be
\EFRT{\text{MYOPIC}} \geq  \frac 1 { \frac{\Delta-\delta}{1-\delta}+1} 
\label{eq:gamma-unif}
\ee
\end{corollary}
Note that theoretically we can scale up the timer rate $\nu$, so that the control phase $T_C$ becomes very small.
\NEW{
 \begin{remark}
 We note that direct application of Theorem~\ref{thm:convert} would yield an efficiency ratio $\approx \frac{1}{\Delta+2}$ as the myopic algorithm obtains $\approx \frac{1}{\Delta+1}$ approximation of the \MWS (To see this note that 
  every link has probabibility $1/(\Delta+1)$ of getting service. Thus all the links of the MWS are included with this probability).
  Therefore Corollary~\ref{myopic-generalresult} also serves as an example in which  more careful analysis can improve the bound of a direct conversion.
 \end{remark}
 }


\section{Analysis Techniques and Proofs}
We provide an overview of the techniques in our proofs.

\textit{Frame Construction.} A key step in the analysis of our scheduling algorithms is a frame construction similar to the one in \cite{tsanikidis2020power}, but based on the joint traffic-fading process. The definition of frame is as follows

\begin{definition}[Frames and Cycles]\label{def:frame}
Starting from an initial traffic and fading state tuple $(\TSTATE{0},\FSTATE{0})=\ve{z}\in \mathcal{Z}$, let $t_i$ denote the $i$-th return time of traffic-fading Markov chain $\TFSTATE{t}$ to $\ve{z}$, $i=1,\cdots$. By convention, define $t_0=0$. The $i$-th cycle $\mathcal{C}_i$ is defined from the beginning of time slot $t_{i-1}+1$ until the end of time slot $t_i$, with cycle length $C_i=t_i-t_{i-1}$. Given a fixed $k\in \mathbb N$, we define the $i$-th frame $\mathcal{F}^{(k)}_i$ as $k$ consecutive cycles $\mathcal{C}_{(i-1)k+1}, \cdots, \mathcal{C}_{ik}$, i.e., from the beginning of slot $t_{(i-1)k}+1$ until the end of slot $t_{ik}$. The length of the $i$-th frame is denoted by $F^{(k)}_i=\sum_{j=(i-1)k+1}^{ik}C_j$. Define $\JFK$ to be the space of all possible  $(\TSTATE{t},\FSTATE{t})$ patterns during a frame $\mathcal{F}^{(k)}$. Note that these patterns start after $\ve{z}$ and end with $\ve{z}$.
\end{definition} 
By the strong Markov property and the positive recurrence of traffic-fading Markov chain $\TFSTATE{t}$, frame lengths $F^{(k)}_i$ are i.i.d with mean $\ES[F^{(k)}]=k\ES[C]$, where $\ES[C]$ is the mean cycle length which is a bounded constant~\cite{dynkin2012theory}. In fact, since state space $\mathcal{Z}$ is finite, all the moments of $C$ (and $F^{(k)}$) are finite.     
We choose a fixed $k$, and, when the context is clear, drop the dependence on $k$ in the notation. 

Define the class of \textit{non-causal $\mathcal{F}$-framed} policies $\strat$ to be the policies that, at the beginning of each frame $\mathcal{F}_i$, have complete information about the traffic-fading pattern in that frame, but have a restriction that they drop the packets that are still in the buffer at the end of the frame. Note that the number of such packets is at most $\MAXDEAD \maxpacks K$, which is negligible compared to the average number of packets in the frame, $\overline a_l\ES[F]=\overline a_lk\ES[C]$, as $k \to \infty$.  Define the rate region 
\be
&\Lambda_{\text NC}(\frm) = \bigcup_{\mu \in \strat}\Lambda _{\mu}.
\ee
Given a policy $\mu \in \strat$, the time-average real-time service rate $\bar{I}_l$ of link $l$ is well defined. By the renewal reward theorem (e.g.~\cite{ross2013applied}, Theorem 5.10), and boundedness of $\ES[F]$, 
\be
\lim_{t \to \infty} \frac{\sum_{s=1}^{t}I_l(s)}{t}=\frac{\ES\left[\sum_{t\in \mathcal{F}}I_l(t)\right]}{\ES[F]}=\bar{I}_l.
\ee
Similarly for the deficit arrival rate $\lambda_l$, defined in \dref{eq:deficitarr}, 
\be \label{eq:lambdarenewal}
\frac{\ES[\sumt \widetilde{a}_l(t)]}{\ES[F]}=\lambda_l,\ l\in \USERS.
\ee
In Definition~\ref{def:frame}, each frame consists of $k$ cycles. Using similar arguments as in \cite{kang2014performance,tsanikidis2020power}, it is easy to see that 
\[\liminf_{k\to \infty}\Lambda_{\text{NC}}(\frm^{(k)}) \supseteq \interior (\Lambda), \]
Hence, if we prove that for a causal policy ALG, there exists a constant $\rho$, and a large $k_0$, such that for all $k\geq k_0$, 
\be \label{eq:region2}
\rho\interior(\Lambda_{\text{NC}}(\frm^{(k)})) \subseteq \Lambda_{\text{ALG}},
\ee
then it follows that $\Lambda_{\text{ALG}} \supseteq \rho \interior(\Lambda)$. For our algorithms, we find a $\rho$ such that \dref{eq:region2} holds for \textit{any traffic-fading process} under our model. Then it follows that $\gamma^\star_{\text{ALG}}\geq \rho$.

\textit{Lyapunov Argument.}
To prove \dref{eq:region2}, we rely on comparing the expected gain of \ALG with that of the non-causal policy that maximizes the expected gain over the frame (max-gain policy). The following proposition, which is similar to that in \cite{tsanikidis2020power}, will be used to prove the main results. We omit its proof, as it is similar to the proof in \cite{tsanikidis2020power} with minor modifications to account for channel uncertainty.  
\begin{proposition}
\label{prop} 
Consider a frame $\frm \equiv \frm^{(k)}$, for a fixed $k$ based on the returns of the traffic-fading process  $\TFSTATE{t}$ to a state $\ve{z}$. Define the norm of initial deficits at the beginning of a frame $\NID{t_0}=\suml w_{l}(t_0)$. Suppose for a causal policy \ALG, given any $\epsilon>0$, there is a $W^\prime$ such that when $\NID{t_0}>W^\prime$,
\be
\frac {\E{\sumt \GALG(t)  | \STATEZ}}{\E{\sumt \GADV(t) | \STATEZ}} \geq \rho-\epsilon,
\ee
where $\STATEZ=\STATETUPLE{t_0}$, and \OPT is the non-causal policy that maximizes the gain over the frame. Then for any $\lambda \in \rho \interior( \Lambda_{NC}(\frm))$, the deficit queues are bounded in the sense of \dref{eq:strong stability}. 
\end{proposition}

\textit{Amortized Gain Analysis.} To use Proposition~\ref{prop}, we need to analyze the achievable gain of \ALG and the non-causal policy \OPT over a frame. Since comparing the gains of the two policies directly is difficult, we adapt an amortized analysis technique from \cite{tsanikidis2020power}, initially extended from~\cite{chin2006online, jez2011one, bienkowski2011randomized,jez2012online}. The general idea is as follows. Let $\STATETUPLE{t}$ be the state under our algorithm at time $t \in \frm$, and $\STATETUPLEP{t}{\mu^{\star}}$ be the state under the optimal policy $\mu^\star$. The traffic-fading process $\TFSTATE{t}=\TFSTATETUPLE{t}$ is identical for both algorithms as it is independent of the actions of the scheduling policy.
 We change the state of $\mu^\star$ (by modifying its buffers and deficits) to make it identical to $\STATETUPLE{t}$, but also give $\mu^\star$ an additional gain that  ensures the change is advantageous for $\mu^\star$ considering the rest of the frame. Let $\GADVA(t)$ denote the amortized gain  of \OPT at time $t$ with any compensated gain, which has the property that 
 \be \label{eq:amorproperty}
 \ES \Big[\sumt \GADVA(t)|J,\STATEZ\Big] \geq  \ES\Big[\sumt \GADV(t)|J,\STATEZ\Big],
 \ee
 given any traffic-fading pattern $J\in \mathcal{J}(\mathcal{F})$ and initial frame state $\STATEZ$.
Then, the following proposition will be useful in bounding the gain and thus the efficiency ratio of our policies.
\begin{proposition} 
\label{prop-approx}
Consider a Markov policy $\ALG$ that for any traffic-fading pattern $J\in \JF$, at any time $t\in \FRAME$, satisfies
\be
\rho \E{\GADVA(t)| J,\STATET} \leq \E{\GALG(t)| J,\STATET} + \EPSG
\label{eq:condition-for-efficiency}
\ee
for some $\EPSG$ which is a measurable function of the frame length $F$, with $\E{F \EPSG}<\infty$. Then  $\EFRT{\ALG} \geq \rho $. 

\label{lemma:general-eff}
\end{proposition}
\begin{proof}
First note that 
$
\E{\GALG(t)|\PAT,\STATET}  =
\E{\GALG(t)|\PAT,\STATET,\STATEZ} 
$
by the Markov property of $\ALG$. Further
$
\E{\GADVA(t)|\PAT,\STATET}  =
\E{\GADVA(t)|\PAT,\STATET,\STATEZ} 
$ since the amortized gain of the max-gain policy does not depend on the past state given the current state and future traffic-fading pattern (note that this amortized gain might depend on policy $\ALG$, which itself is Markov). Hence, taking expectation of both sides of \dref{eq:condition-for-efficiency}, conditional on $\PAT, \STATEZ$, and using the law of iterated expectations, we get
\[
\rho \E{\GADVA(t)|\PAT,\STATEZ} \leq 
\E{\GALG(t)|\PAT,\STATEZ} +\E{\EPSG | \PAT,\STATEZ}
\]
Summing over the frame, and using the fact that $\E{\EPSG | \PAT,\STATEZ}=\EPSG$ since $F$ is determined by $J$, we have
\[
\rho \E{\sumt \GADVA(t)|\PAT,\STATEZ} \leq 
\E{\sumt \GALG(t)|\PAT,\STATEZ} +F \EPSG.
\]
Using the definition~\dref{eq:amorproperty}, and taking expectations over the randomness of the traffic-fading pattern, we obtain
\ben
\rho \E{\sumt \GADV(t)|\STATEZ} \leq 
\E{\sumt \GALG(t)|\STATEZ} + \E{F \EPSG}.
\een
Note that by assumption, $\E{F \EPSG} < \infty$. If we show that 
$\lim_{\NID{t_0}\rightarrow \infty}\E{\sumt \GADV(t)|\STATEZ} = \infty$, then
$$ \rho \leq 
\lim_{\NID{t_0}\rightarrow \infty} \frac{ \E{\sumt \GALG(t)|\STATEZ}} 
{
\E{\sumt \GADV(t)|\STATEZ}
}, $$ and using Proposition \ref{prop}, we obtain that $\gamma_{\mu}\geq \rho$.

To that end, consider the link $l_1 = \argmax_{l}  w_{l}(t_0)$. Since $\mu^\star$ maximizes the expected gain over the frame, its expected gain cannot be lower than a policy that simply schedules link $l_1$ whenever it has available packets. 
Choose $\FRAME \equiv \FRAMEK$ with any $k\geq2$. Then by the non-trivial traffic-fading Markov chain assumption (Section~\ref{sec:model}), there is a traffic-fading pattern $J^\prime$ of length $F_{J^\prime}$ with some nonzero probability of occurring $\Pr{(J^\prime)}>0$, in which a packet arrives for link $l_1$ at some time $t_1 \in \FRAME$, and at a later time $t_2 \in \FRAME$ before the packet expires, we have $\FP_l(t_2)>0$, i.e., $l_1\in \BUFFERT{t_2}$ and $\E{\IF_{l_1}(t_2)}=\FP_{l_1}(t_2)>0$. Then trivially
\begin{flalign*}
\E{{\sumt} \GADV(t)|\STATEZ} &\geq 
\E { w_{l_1}(t_2) \IF_{l_1}(t_2) | {J^\prime},\STATEZ} \Pr{(J^\prime)}  \\
& \geq(w_{l_1}(t_0)-F_{J^\prime}) \FP_{l_1}(t_2) \Pr{(J^\prime)}  \\
&\geq (\frac{\NID{t_0}}{K}-F_{J^\prime}) 
\FP_{l_1}(t_2) \Pr{(J^\prime)},
\end{flalign*}
which shows
$\lim_{\NID{t_0}\rightarrow \infty}\E{\sumt \GADV(t)|\STATEZ} = \infty$.
\end{proof}
The following lemma describes a generic amortized gain computation for general networks that allows us to modify the state of the max-gain policy during a frame to match the state of the considered policy \ALG. Recall from~\dref{def:phi} that $\PHI_l(t)$ is the probability that link $l$ is scheduled under \ALG, and $\PI_M(t)$ is the probability that schedule $\M{}$ is selected. Below we drop their dependence on time to simplify the notation. 
\begin{lemma} \label{lemmageneral-gain-analysis} 
For any pattern $J\in \JF$ in a frame $\FRAME$, given a {Markov} policy \ALG,
the amortized gain  of the max-gain policy \OPT (w.r.t. \ALG) if it selects maximal schedule $\M{}\in \AMIS$ at time $t$ in the frame, is given by 
\ben
\label{eq-amort1}
\ERF{\GADVE{\M{}}|J,\STATET}=
\W{}+
\sum_{l \not\in \M{}} w_l(t) \PHI_l  \FP_l(t)
+\EPSM \\
\leq \W{} (1-\PI_{M}) + \ERF{\GALG(t)|\STATET} + \EPSM \label{eq-amort2},
\een
where $\EPSM:=K(F+\maxpacks \MAXDEAD)$. 
\end{lemma}
\begin{proof}
The main idea is similar to the one in \cite{tsanikidis2020power}, but we have to account for fading. Suppose \OPT attempts to transmit the earliest-deadline packets of links from schedule $\M{}$ whereas \ALG attempts the earliest-deadline packets from schedule $\M{\prime}$. 
We need to modify the state of \OPT, i.e., the buffers and deficits, so it is  identical with the state of \ALG. To achieve this, we allow \OPT to additionally transmit the packets of links  successfully transmitted by \ALG but not \OPT, i.e, $S_1=(\MT{\prime}\setminus \MT{}) \cap \ACTIVES{t}$. Transmitting such packets for \OPT might not be advantageous for its total gain as the weight of these packets can increase by $\maxpacks \MAXDEAD$ before they could be transmitted. Thus giving an additional reward $K\maxpacks \MAXDEAD$ to \OPT guarantees that the modification is advantageous. Further, we insert the packets transmitted successfully by \OPT but not \ALG, i.e., $S_2=(\MT{}\setminus\MT{\prime})\cap \ACTIVES{t}$, back to its buffers (which is advantageous for \OPT). Further to make the deficits identical, we increase the deficit counters of \OPT for links in $S_2$, which is advantageous for the total-gain within the frame. Additionally, we decrease the deficit counters of \OPT for the links in $S_1$. This change might not be advantageous for the total gain of \OPT, thus we give it extra reward for every possible subsequent transmission over links, which is at most $KF$. Hence the total additional compensation is $\EPSM=K(F+ \maxpacks \MAXDEAD)$

Following the above argument, the expected amortized gain of \OPT, when it selects $\M{}$, is bounded as
\begin{flalign}
& \ERF{\GADVE{\M{}}|J,\STATET} = 
\sum_{l \in \M{}} \FP_{l}(t) w_l(t) +    \nonumber \\
&\sum_{\M{\prime} \in \AMIS\setminus \{\M{}\}} \PI_{\M{\prime}}(t) \sum_{l \in \M{\prime} \setminus \M{}} \FP_{l}(t) w_l(t) +\EPSM  \nonumber \\
&\overset{(a)}{=} \W{}+
\sum_{l \not\in \M{}} w_l(t) \PHI_l \FP_{l}(t) +\EPSM \label{eq:intermed2}
\end{flalign}
where $(a)$ follows from definitions of $\PHI_l$ and $\W{}$. The inequality in the lemma's statement follows by noting that  
$$\mbox{\dref{eq:intermed2}}\leq \W{} +
\sum_{\M{\prime}\in \AMIS \setminus \{\M{}\}} \W{\prime} \PI_{\M{\prime}} + \EPSM,
$$
since $\sum_{l \in \M{\prime} \setminus \M{}} \FP_{l}(t) w_l(t)\leq  \W{\prime}$, and by using~\dref{eq:glag-generic2}.
%
\end{proof}

\subsection{Proof of Theorem \ref{thm:convert}: Conversion Result}
In the rest of proofs, for notional compactness, we define $\ETT{\cdot}:=\E{\cdot|\STATET,J}$ and $\EOT{\cdot}:=\E{\cdot|\STATET}$. Now consider a pattern $J\in \JF$. When $\NID{t}\geq W^\prime$, the amortized gain of $\mu^\star$, if it selects schedule $\M{}$, is bounded as
\begin{flalign}
\ERFTJ{\GADVE{\M{}}} & \overset{(a)}{\leq} \W{}(1-\PI_M) + \ERFT{\GALG(t)} + \EPSM  \nonumber\\
&\leq \W{\star} + \ERFT{\GALG(t)} + \EPSM  \nonumber \\
& = \ERFT{\GALG(t)}(\frac {\W{*}} { \ERFT{\GALG(t)} } + 1 ) + \EPSM  \nonumber\\
&\overset{(b)}{\leq} {\ERFT{\GALG}}  (1/\psi + 1) + \EPSM, \nonumber
\end{flalign}
where in $(a)$ we used  Lemma~\ref{lemmageneral-gain-analysis}, and in $(b)$ we used the main assumption that $\ALG$ obtains $\psi$ fraction of the maximum weight schedule.
This inequality does not depend on the particular choice $\M{}$.
Similarly in the case that $\NID{t} \leq W^\prime$, it can be seen from \dref{eq:galg-generic} and Lemma~\ref{lemmageneral-gain-analysis} that $\ERFTJ{\GADVE{\M{}}} \leq 2 W^\prime+\EPSM$. Consequently in either case, $\ERFTJ{\GADVE{\M{}}} \leq 2 W^\prime+\EPSM+\ERFT{\GALG(t)}(1/\psi + 1)$.
Applying Proposition~\ref{lemma:general-eff} with $\rho=(1/\psi + 1)^{-1}$, we obtain the result. 

\subsection{Analysis of \MPALGONAME: Proof of Theorem \ref{main-theorem-general}.} 
Using Lemma~\ref{lemmageneral-gain-analysis}, and probabilities of FAMIX-MS indicated in \dref{eq:mwisprob}, in both cases of $i\leq \nstar$ and $i> \nstar$, the amortized gain of \OPT can be bounded as
\begin{flalign*}
\ERFTJ{\GADVE{M_i}} &\leq  
\W{(i)} (1-\PI_{i}) +
\ERFT{\GALG(t)} + 
\EPSM   \\ &\leq  
\subhrm{\nstar}+
\ERFT{\GALG(t)} + 
\EPSM.
\end{flalign*}

%
%
\noindent
Thus regardless of the schedule selected by \OPT, we have
\[
\ERFTJ{\GADVA(t)} \leq  
\ERFT{\GALG(t)} \Big(
\frac { \subhrm{\nstar}} {
\ERFT{\GALG(t)}
} +1\Big) +  \EPSM
\]
Note that $\CAMISE \geq \CAMIS \geq \nstar$, hence it suffices to show that
\be
\frac { \subhrm{\nstar}} {
\ERFT{\GALG(t)}
} \leq 
\frac{\nstar-1}{\nstar} ,\label{eq:ineq-to-show}
\ee
as then by Proposition \ref{prop-approx}
it will follow that $\EFRT{\MPALGONAME} \geq (\frac{\CAMISE-1}{\CAMISE}+1)^{-1}= \GRATIO$. To show \dref{eq:ineq-to-show}, note that 
\ben
\frac { \subhrm{\nstar}} {
\ERFT{\GALG(t)}
} = \frac{\nstar \subhrm{\nstar} - (\nstar-1) \ERFT{\GALG(t)}}{
\nstar \ERFT{\GALG(t)}
} + \frac {\nstar-1}{\nstar}
\een
and $\nstar \subhrm{\nstar} - (\nstar-1) \ERFT{\GALG(t)}\leq 0$. This inequality holds because by \dref{eq:glag-generic2} and \dref{eq:mwisprob}, it follows that $\ERFT{\GALG(t)} =\sum_{j=1}^{\nstar}\W{(j)} - \nstar \subhrm{\nstar}$ and hence the inequality becomes 
\ben
&\subhrm{\nstar} \nstar^2  \leq 
(\nstar-1) \sum_{j=1}^{\nstar}\W{(j)},
\een
which holds by \dref{def:subhrm} and the arithmetic mean-harmonic mean inequality applied to $\{\W{(i)},i=1,...,\nstar\}$.

\subsection{Analysis of \ALGONAME: Proof of Theorem~\ref{main-theorem-collocated}}
We first state the following Lemma that allows us to focus on policies that transmit from non-dominated links.
\begin{lemma} \label{nd-lemma}
Given $J\in \JF$, let $\ADVR$ be the maximum-gain policy that transmits only from non-dominated links at any time (we refer to $\ADVR$ as max-gain ND-policy) and the maximum-gain policy $\OPT$ that can transmit any packet, then
\begin{flalign*}
\textstyle \E{\sumt \GADVND(t) | \STATEZ,J} \geq& \textstyle \E{\sumt \GADV(t) | \STATEZ,J} \\
& - (\maxpacks+1) F^2.
\end{flalign*}
\end{lemma}
\begin{proof}
First we consider the case where deficits do not vary over time regardless of arrivals or transmissions. We argue that in this case transmitting from a non-dominated link has always higher total expected gain from any state considering the remaining frame for any pattern $J$. 
Let  $\VAL{l}{\BSTATE{}}{t}$ denote the maximum expected gain over the remaining frame for the given pattern $J$, if at the current time slot $t\in \FRAME$ with buffers $\BSTATE{}$, link $l$ is scheduled. Further, define $\VAL{}{\BSTATE{}}{t} = \MAX{l} {\VAL{l}{\BSTATE{}}{t}}$. Finally define $\VAL{ND}{S}{t}$ to be the maximum expected gain for the remaining frame over policies that only schedule non-dominated links. Consider the earliest-deadline packets $\PA_u,\PA_v$ of links $u,v$, with $u$ dominating $v$. We argue that transmitting  $\PA_u$ yields a higher expected gain than that of $\PA_v$. If $\PA_u$ is scheduled we have:
\begin{flalign*}
\VAL{u}{\BSTATEE}{t}  = &
q (w_u + \VAL{}{\BSTATEE^\prime_{(\PA_u)}  }{t+1}) \\ 
&+(1-q) \VAL{}{\BSTATEE^\prime}{t+1},
\end{flalign*}
where $\BSTATEE^\prime$ are the updated buffers due to regular buffer dynamics but ignoring the scheduled packet by our policy, and $\BSTATEE^\prime_{(\PA)}$ indicates buffer $\BSTATEE^\prime$ without packet $\PA$. Similar expression can be written if $\PA_v$ is scheduled. Hence, to show $\VAL{u}{\BSTATEE}{t} \geq 
\VAL{v}{\BSTATEE}{t}$, equivalently we can show
\be \label{eq:intermediate}
w_u + \VAL{}{\BSTATEE^\prime_{(\PA_u)}}{t+1}
\geq w_v + \VAL{}{\BSTATEE^\prime_{(\PA_v)}}{t+1}.
\ee
Now note that by the domination definition, the deadline of packet $v$ is at least as long as that of $u$, hence the optimal policy $\mu$ for
$\VAL{}{\BSTATEE^\prime_{(\PA_v)}}{t+1}$ that can attempt packet $u$ can be used to construct a policy $\mu^\prime$ for 
$\VAL{}{\BSTATEE^\prime_{(\PA_u)}}{t+1}$ that attempts $v$ instead of $u$ whenever $\mu$ would have scheduled $\PA_u$. This is possible since $\PA_v$ does not expire before $\PA_u$. 
Now consider a coupling for the channel outcomes in $\mu$ and $\mu^\prime$ such that channel $C^\mu_u(t)$ under $\mu$ is identical to $C^{\mu^\prime}_v(t)$ under $\mu^\prime$. This does not affect the marginal success probability of links as both links $u,v$ have the same probability of success $\FP$. Further each success when transmitting $v$ under $\mu^\prime$ results in reward $w_v$, instead of $w_u$ under $\mu$, therefore the gain of policy $\mu^\prime$ is stochastically larger than that of policy $\mu$, and hence in the expectation, 
\[
\VAL{}{\BSTATEE^\prime_{(\PA_u)}}{t+1}
\geq (w_v-w_u) + \VAL{}{\BSTATEE^\prime_{(\PA_v)}}{t+1}
\]
Thus in the fixed-deficit case we have shown that $\VAL{ND}{\BSTATEE}{t}\geq \VAL{}{\BSTATEE}{t}$.

Now consider the case with time-varying deficits, and let $\VALV{\bcdot}{\BSTATEE}{t}$ denote the corresponding quantities. In this notation $\VALV{}{\BSTATE{t_0}}{t_0} = \E{\sumt \GADV(t)|J,\STATEZ}$. If pattern $J$ has length $F$, then (a): 
$\VAL{}{\BSTATEE}{t} \geq \VALV{}{\BSTATEE}{t}-\maxpacks F^2$ since any transmission under any policy in the varying-deficit case will yield a gain of at most $\maxpacks F$ higher than the fixed-deficit case, since the deficit cannot increase more than $\maxpacks F$ during a frame, and we can have at most $F$ such transmissions. Similarly we obtain (b): $\VALV{ND}{\BSTATEE}{t}+ F^2 \geq \VAL{ND}{\BSTATEE}{t}$, as in the time-varying deficit case every transmission can result in at most $F$ gain less than the fixed-deficit case, and we can have at most $F$ such transmissions. 
By (a) and (b), we obtain
\ben
\VALV{ND}{\BSTATEE}{t} \geq \VALV{}{\BSTATEE}{t} - (\maxpacks+1) F^2.
\label{eq:vnd-vs-v}
\een 
\end{proof}
\textit{Proof of Theorem \ref{main-theorem-collocated}.}
In view of Lemma~\ref{nd-lemma}, we perform an amortized gain analysis for \ALGONAME in comparison with ND-policies. 
Suppose \ALGONAME\ decides to schedule the earliest-deadline packet $\PA_f = (w_f,e_f)$ of link $f$ and ND-policy $\ADVR$ transmits a packet $\PA_z = (w_z,e_z)$ from a different non-dominated link $z$ ($z\not=f$). The state of $\ADVR$ and \ALGONAME will be different
in the following four cases, 
\begin{enumerate}[leftmargin=*]
\item $\ACTIVELINK{f}{t}, \INACTIVELINK{z}{t}$: In this case,
we make the buffers of the two identical by allowing $\ADVR$ to also transmit $\PA_f$, and give extra reward $\maxpacks \MAXDEAD$. Further we reduce the deficit in $\ADVR$ for link $f$ by $1$ and we give reward $F$ to $\ADVR$ similarly as in the proof of  Lemma~\ref{lemmageneral-gain-analysis}.
\item $\INACTIVELINK{f}{t}, \ACTIVELINK{z}{t}$: We replace $\PA_z$ in the buffer of $\ADVR$ and increase the deficit of link $z$ in $\ADVR$ which are advantageous for $\ADVR$.
\item $e_f \leq e_z, w_f \leq w_z,$ and $\ACTIVELINK{f}{t}, \ACTIVELINK{z}{t}$: We replace $\PA_f$ from link $f$ with $\PA_z$ in link $z$ in buffers of $\ADVR$. Packet $\PA_z$ has higher deadline and higher weight at time $t$. 
Since both packets will expire in at most $\MAXDEAD$ slots, the deficit of $f$ can only increase by at most $\MAXDEAD \maxpacks$ before $\PA_f$ expires, whereas the deficit of $z$ can decrease by at most $\MAXDEAD$. Therefore giving $\ADVR$ additional compensation of $(1+\maxpacks) \MAXDEAD$ will guarantee that the modification is advantageous. 
Further, we decrease the deficit of link $f$ by one ($w_f-1$ in $\ADVR$) and we increase the deficit of link $z$ by one ($w_z+1$ in $\ADVR$). 
     To compensate for the decrease in deficit of link $f$ we give $\ADVR$ extra gain of $F$. Hence, the total compensation is bounded by   $F+(a_{max}+1)\MAXDEAD.$

\item $e_z \leq e_f, w_z \leq w_f$ and $\ACTIVELINK{f}{t}, \ACTIVELINK{z}{t}$: In this case, we allow $\ADVR$ to additionally transmit packet $\PA_f$ at time $t$ and give extra reward $\maxpacks \MAXDEAD$, and inject a copy of packet $\PA_z$ to the buffer of link $z$. This makes the buffers identical, but results in the decrease of deficit of link $f$ by one.
    To guarantee the change is advantageous for $\ADVR$ considering the rest of the frame, we give it extra reward $F$.
\end{enumerate}
In all cases, the additional compensation is bounded by $\EPSC=F+(\maxpacks+1)\MAXDEAD$. Let $\FP_{ij}=\Pr{(C_i(t)=1,C_j(t)=1)}\geq q^2$, or equivalently $\FP_{ij}=q^2+\epsilon_{ij}$ for some $\epsilon_{ij}\geq0$, by the positive correlation assumption, and hence $\Pr{(C_i(t)=1,C_j(t)=0)}= \FP- \FP_{ij} = q(1-q)-\epsilon_{ij}$. Then,
\ben
\ETT{\GADVEND{i}} \leq {\FP{} w_i}
+
\sum_{j} \pi_j (\FP-\FP_{ij}) w_j
+ 
{\sum_{j < i } w_j\pi_j \FP_{ij}}
\\+ \EPSC =  
{\FP{} w_i}
+
(1-q) \sum_{j} \pi_j \FP w_j
+ 
q{\sum_{j < i } w_j\pi_j q}
\\
-\sum_{j} \pi_j \epsilon_{ij} w_j 
+\sum_{j<i} \pi_j \epsilon_{ij} w_j 
+\EPSC 
\\
\leq \FP{} [w_i + (1-\FP{})\sum_{j} \pi_j w_j  + q \sum_{j < i} w_j \pi_j ] + \EPSC
\een
Now notice that for 
$\pi_i \leq \frac 1 {\FP{}} (1- \frac{w_{i+1}}{w_i})$ the right-hand-side of the above inequality is maximized for $i=1$. This can be verified by computing the difference of the gains for choices $i,i+1$. Hence, we have
\be
\ETT{\GADVAND(t)} \leq 
\max_i \ETT{\GADVEND{i}} = 
\ETT{\GADVEND{1}} \nonumber\\=
\FP{} w_1 + (1-\FP{}) \EOT{\GALG(t)} + \EPSC. \label{eq:boundadvc}
\ee
Let $\nstar$ be the number of links with positive probability in \dref{eq:ndprob}. For the gain of \ALGONAME, we have 
\begin{flalign*}
&\EOT{\GALG(t)} =
\sum_{i=1}^{\nstar} q \pi_i w_i  \\
&= \sum_{i=1}^{\nstar-1} (w_i - w_{i+1}) +
 \FP{} (1 - \sum_{i=1}^{\nstar-1} \pi_i)  w_{\nstar}\\
& =w_1+ w_{\nstar} (-1  
+ \FP{} - \FP{}\sum_{i=1}^{\nstar-1} \pi_i) 
\\
& \overset{(a)}{=}w_1 \left(
1- 
(1  
- \FP{} + \FP{}\sum_{i=1}^{\nstar-1} \pi_i )
\prod_{i=1}^{\nstar-1} (1-\pi_i \FP{})
\right)  
\\
& \overset{(b)}{\geq} w_1 (1-\left(\frac{\nstar - \FP{}}{\nstar}\right)^{\nstar}) \geq  w_1 (1-e^{-\FP{}}),
\end{flalign*}
where in $(a)$, by the definition of $\pi_i$ in \dref{eq:ndprob}, we used: 
$w_{\nstar}= w_1 \prod_{i=1}^{\nstar-1} (1-\pi_i \FP{})$. In $(b)$ we used the geometric-arithmetic inequality. 
Using the above relation and \dref{eq:boundadvc} we get 
\[
\ETT{\GADVAND(t)}
\leq 
\EOT{\GALG(t)} (
\frac{q}{1-e^{-q}} + (1-q)) + \EPSC
\]
Using similar arguments as in the proof of Proposition \ref{lemma:general-eff}, we can sum over the entire frame to obtain
\begin{flalign*}
&\textstyle \E{\sumt \GADVAND(t)|J,\STATEZ} \leq \\
&\textstyle \E{\sumt \GALG(t)|J,\STATEZ}
(\frac{q}{1-e^{-q}} + 1-q)
+F \EPSC.
\end{flalign*}
Due to the amortized analysis, \dref{eq:amorproperty} holds. Then by  
using Lemma~\ref{nd-lemma} we obtain
\ben
\textstyle \E{\sumt \GADV(t)|J,\STATEZ} - F^2 (\maxpacks+1) \leq  \\
\textstyle \E{\sumt \GALG(t)|J,\STATEZ}
(\frac{\FP{}}{1-e^{-\FP{}}} + 1-\FP{})
+F \EPSC.
\een
Let $\rho = ({\frac{\FP{}}{1-e^{-\FP{}}} + 1-\FP{}})^{-1}$. After taking expectation with respect to the randomness of the frame, rearranging terms and dividing appropriately, we obtain
\ben
\rho 
 \leq 
\frac{ \E{\sumt \GALG(t)|\STATEZ} }{
\E{\sumt \GADV(t) | \STATEZ}
}
+ \frac{ \E{F \EPSC + F^2 (\maxpacks+1)} }{
\E{\sumt \GADV(t) |\STATEZ} } \rho
\een
The proof is concluded as in the proof of Proposition \ref{lemma:general-eff}, by 
$\lim_{\NID{t_0}\rightarrow \infty}\E{\sumt \GADV(t) | \STATEZ} = \ \infty$, and applying Proposition~\ref{prop}.

\subsection{Proof of Proposition \ref{coloring-proposition}: Low-Complexity Variants}

Suppose that under \ALG
we have 
\be \label{eq:equivineq}
 \psi(t) \max_{M \in \AMIS} \ETT{\GADVEM{}}\leq \max_{M \in \AMIS_0(t)} {\ETT {\GADVEM{}}}.
\ee
Since \ALG randomizes according to the probabilities of \MPALGONAME over the maximal schedules in $\AMIS_0(t)$, we have 
\be \label{eq:xiieq}
\xi(t) \max_{M \in \AMIS_0(t)} \{ \ETT{\GADVEM{}} \} \leq { \EOT{\GALG(t)}}+\EPS,
\ee
where $\xi(t) := \frac {|\AMIS_0(t)|}{2 |\AMIS_0(t)|-1}$, by using the same arguments as in the proof of Theorem \ref{main-theorem-general} applied to $\AMIS_0(t)$. By \dref{eq:equivineq} and \dref{eq:xiieq}, 
$$
\MIN{t\geq 0} { \xi(t) \psi(t)} \ETT {\GADVA(t) } \leq  \EOT {\GALG(t)}+\EPS.
$$
Then by Proposition \ref{prop-approx}, it follows that $\gamma \geq
\MIN{t\geq 0} {\psi(t) 
 \xi(t) 
}. 
$
Hence to conclude the proof, we simply need to argue that \dref{eq:colortricksuffnew} implies \dref{eq:equivineq} 
Using Lemma~\ref{lemmageneral-gain-analysis}, and \dref{eq:galg-generic}, it can be shown that
\ben
\ERFTJ{\GADVE{\M{}}}=
\WE{} \bar \PHI_l +
\ERFT{\GALG(t)}
+\EPSM \label{eq:clr-equiv}
\een
Plugging the above expression in \dref{eq:equivineq}, it can be seen that
\dref{eq:colortricksuffnew} is a sufficient condition for \dref{eq:equivineq} to hold. 

\subsection{Proof of Corollary~\ref{myopic-generalresult}: Myopic Distributed Algorithm} \label{sec:unif-analysis}
The amortized gain can be written as
\begin{flalign}
&\ETT{\GADVEM{}} =\sum_{l \in \MT{}} \FP_l(t) w_l(t)  \bar{\PHI}_l + \E{\GALG(t)}+\EPSM \nonumber\\ 
&\overset{(a)}{=} \ETT{\GALG(t)} ( 
\frac { \sum_{l \in \MT{}} \FP_l(t) w_l(t) \bar{\PHI}_l } { \sum_{l} \PHI_l \FP_l(t) w_l(t) } +1)+ \EPSM  \label{eq:distr-factored}
\end{flalign} 
where in (a), we used \dref{eq:galg-generic}. 

Let $A_l=\bigcap_{z \in N(l)}\{T_l < T_z \}$ be  the event that the timer of $l$ expires earlier than that of its neighbors.
 We can obtain a bound on the probability of link getting scheduled as follows: 
\be
\phi_l & \geq & 
\Pr{[A_l \cap \{T_l < T_C\}]} = \Pr{[A_l]}\Pr{[T_l < T_C|A_l]} \nonumber \\
&\overset{(a)}{\geq} & 
 \frac{\nu_l }{\nu_l + \sum_{z \in N(z)} \nu_{z}} (
1 - e^{-(\nu_l + \sum_{z \in N(z)} \nu_{z})T_C}) \nonumber  \\
&\geq &  \frac{\nu_l }{\nu_l + \sum_{z \in N(z)} \nu_{z}} (
1 - e^{-\nu_l T_C}) \nonumber  \\
&\overset{(b)}{\geq} &  \frac{\nu_l }{\nu_l + \sum_{z \in N(z)} \nu_{z}} (
1 - \delta), \label{eq:phimyopic}
\ee
where (a) follows from the fact that $T_l|A_l \sim Exp(\nu_l + \sum_{z \in N(z)} \nu_{z})$ (easy to verify),  and (b) from the choice of $T_C$.
Hence, using~\dref{eq:phimyopic}, in the case that $\nu_l=\nu$, we get $\PHI_l \geq \frac {1-\delta}{1+\Delta}$. Hence,
\ben
\frac { \sum_{l \in \MT{}} \FP_l(t) w_l(t) ( 1- \PHI_l) } { \sum_{l} \PHI_l \FP_l(t) w_l(t) } \leq  
\frac { \sum_{l \in \MT{}} \FP_l(t) w_l(t) ( 1- \frac{1-\delta}{\Delta+1}) } { \sum_{l} \FP_l(t) w_l(t) \frac{1-\delta}{\Delta+1} }  \\ 
\leq \frac { \sum_{l \in \MT{}} \FP_l(t) w_l(t) ( 1- \frac{1-\delta}{\Delta+1}) } { \sum_{l \in \MT{}} \FP_l(t) w_l(t) \frac{1-\delta}{\Delta+1}} 
\leq
\frac{ \Delta-\delta } {1-\delta},
\een
where the first inequality is due to the fact that the ratio is a decreasing function of $\phi_l$, thus using \dref{eq:distr-factored}, 
\[
\ETT{ \GADVE{M}} \leq \ETT{\GALG(t)} (\frac{ \Delta - \delta} {1-\delta}+1) + \EPSM ,
\]
where $\delta$ can be made arbitrarily small. Using Proposition \ref{prop-approx}, we obtain the final result.

\section{Simulation Results}
We performed simulations under several networks and traffic-fading scenarios. Our algorithms \ALGONAME and \MPALGONAME can considerably outperform \GMS and \MWS. We present two of the simulations here due to space constraint.

First, we consider the traffic of Pattern B of Figure~\ref{fig-traffic-fading}, but with  i.i.d. channel  success probability $q$. The results are shown in Figure~\ref{fig:deterministic-collocated} for $q=0.6$ and equal target delivery ratio $p$ for links. Note that the optimal cannot achieve $p>0.6$, hence \ALGONAME is near optimal. We observed a similar behavior for other values of $q$ and different number of links. 

Next, consider a network with 5 links and $G_I$ with edges $\{(l_1,l_2),(l_2,l_3),(l_2,l_4),(l_4,l_5)\}$. The traffic-fading process for $\{l_1,l_3,l_4\}$ is as in link 2 of Pattern B in Figure~\ref{fig-traffic-fading}, and for links $\{l_2,l_5\}$ is as in link 1. We set an equal target delivery ratio $p$ for all the links. Figure~\ref{fig:simple-graph-sub2-general} shows the results. We see \MPALGONAME can support a significantly higher delivery ratio than other algorithms. In this case, optimal cannot achieve $p>0.75$.

Next, consider a simple network with links $\USERS=\{l_1,l_2,l_3\}$ and edges $\{(l_1,l_2),(l_1,l_3)\}$ in the interference graph $G_I$ with Pattern A from Figure~\ref{fig-traffic-fading}.
Then \GMS becomes unstable for  $\mathbf{p}=(\frac{1}{3}+\epsilon,\frac{1}{2}+\epsilon,\frac{1}{2}+\epsilon)$, whereas the optimal can satisfy $\mathbf{p}=(\frac {11}{12}, \frac{5}{6},\frac{5}{6})$.
\NEW{This example illustrates the existence of simple examples with non-equal delivery ratio requirements where \GMS performs poorly. Similar examples were constructed for more links.}


\begin{figure}[t]
\centering
\begin{subfigure}{.25\textwidth}
  \centering
  \includegraphics[width=0.95\linewidth]{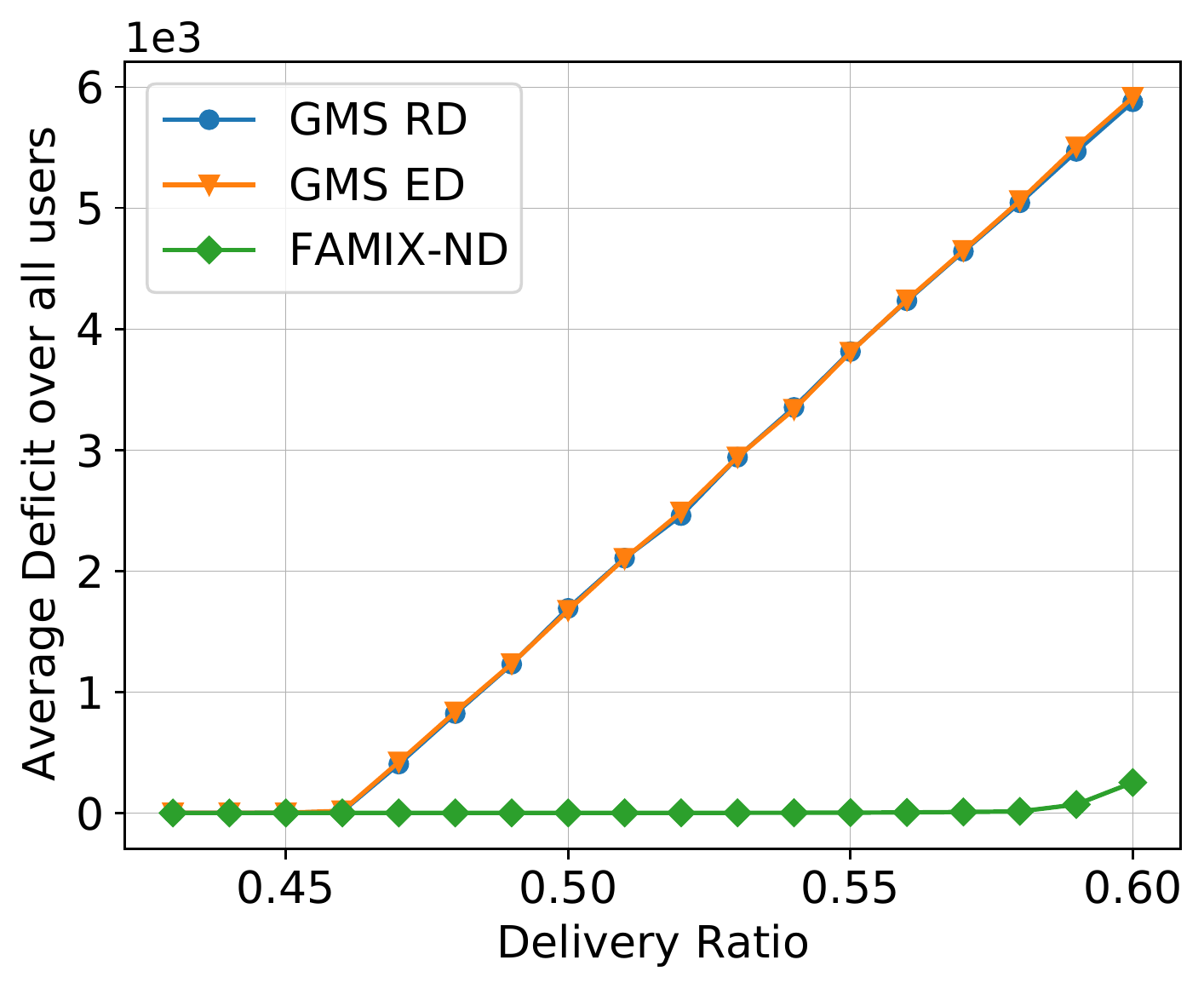}
  \caption{2-link collocated network.}
  \label{fig:deterministic-collocated}
\end{subfigure}%
\begin{subfigure}{.25\textwidth}
  \centering
  \includegraphics[width=0.95\linewidth]{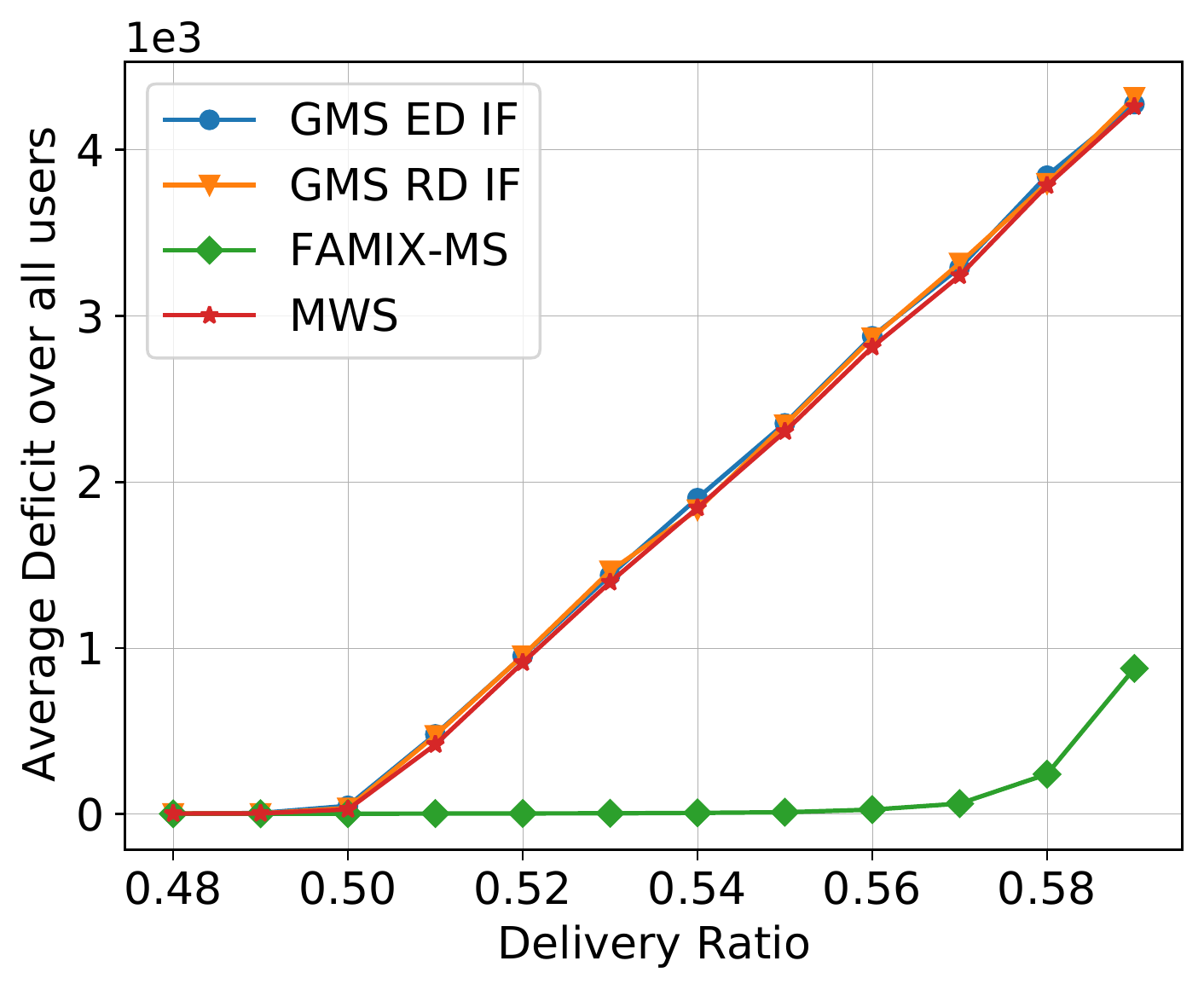}
  \caption{5-link sparse network}
  \label{fig:simple-graph-sub2-general}
\end{subfigure}
\caption{Performance comparison of various algorithms.}
\vspace{- 0.1 in}
\end{figure}

%
\section{Conclusion}
We considered scheduling of real-time traffic over fading channels, where traffic (arrival and deadline) and links' reliability evolves as an \textit{unknown} finite-state Markov chain. We provided a conversion result that shows classical non-real-time scheduling algorithms like \MWS and \GMS can be ported to this setting and characterized their efficiency ratio. We then extended the randomized algorithms from \cite{tsanikidis2020power} to fading channels. We further proposed low-complexity and myopic distributed randomized algorithms and characterized their efficiency ratio. Investigating more efficient low-complexity  and distributed randomized algorithms could be an interesting future research.

\newpage
\bibliographystyle{IEEEtran}
\bibliography{citations}

\begin{thebibliography}{10}
\providecommand{\url}[1]{#1}
\csname url@samestyle\endcsname
\providecommand{\newblock}{\relax}
\providecommand{\bibinfo}[2]{#2}
\providecommand{\BIBentrySTDinterwordspacing}{\spaceskip=0pt\relax}
\providecommand{\BIBentryALTinterwordstretchfactor}{4}
\providecommand{\BIBentryALTinterwordspacing}{\spaceskip=\fontdimen2\font plus
\BIBentryALTinterwordstretchfactor\fontdimen3\font minus
  \fontdimen4\font\relax}
\providecommand{\BIBforeignlanguage}[2]{{%
\expandafter\ifx\csname l@#1\endcsname\relax
\typeout{** WARNING: IEEEtran.bst: No hyphenation pattern has been}%
\typeout{** loaded for the language `#1'. Using the pattern for}%
\typeout{** the default language instead.}%
\else
\language=\csname l@#1\endcsname
\fi
#2}}
\providecommand{\BIBdecl}{\relax}
\BIBdecl

\bibitem{tassiulas1990stability}
L.~Tassiulas and A.~Ephremides, ``Stability properties of constrained queueing
  systems and scheduling policies for maximum throughput in multihop radio
  networks,'' in \emph{29th IEEE Conference on Decision and Control}, 1990, pp.
  2130--2132.

\bibitem{joo2009understanding}
C.~Joo, X.~Lin, and N.~B. Shroff, ``Understanding the capacity region of the
  greedy maximal scheduling algorithm in multihop wireless networks,''
  \emph{IEEE/ACM Transactions on Networking (TON)}, vol.~17, no.~4, pp.
  1132--1145, 2009.

\bibitem{dimakis2006sufficient}
A.~Dimakis and J.~Walrand, ``Sufficient conditions for stability of
  longest-queue-first scheduling: Second-order properties using fluid limits,''
  \emph{Advances in Applied probability}, vol.~38, no.~2, pp. 505--521, 2006.

\bibitem{ghaderi2010design}
J.~Ghaderi and R.~Srikant, ``On the design of efficient {CSMA} algorithms for
  wireless networks,'' in \emph{49th IEEE Conference on Decision and Control
  (CDC)}.\hskip 1em plus 0.5em minus 0.4em\relax IEEE, 2010, pp. 954--959.

\bibitem{ni2012q}
J.~Ni, B.~Tan, and R.~Srikant, ``Q-{CSMA}: Queue-length-based {CSMA/CA}
  algorithms for achieving maximum throughput and low delay in wireless
  networks,'' \emph{IEEE/ACM Transactions on Networking (ToN)}, vol.~20, no.~3,
  pp. 825--836, 2012.

\bibitem{shah2010delay}
D.~Shah and J.~Shin, ``Delay optimal queue-based {CSMA},'' in \emph{ACM
  SIGMETRICS Performance Evaluation Review}, vol.~38, no.~1.\hskip 1em plus
  0.5em minus 0.4em\relax ACM, 2010, pp. 373--374.

\bibitem{lu2015real}
C.~Lu, A.~Saifullah, B.~Li, M.~Sha, H.~Gonzalez, D.~Gunatilaka, C.~Wu, L.~Nie,
  and Y.~Chen, ``Real-time wireless sensor-actuator networks for industrial
  cyber-physical systems,'' \emph{Proceedings of the IEEE}, vol. 104, no.~5,
  pp. 1013--1024, 2015.

\bibitem{song2008wirelesshart}
J.~Song, S.~Han, A.~Mok, D.~Chen, M.~Lucas, M.~Nixon, and W.~Pratt,
  ``Wirelesshart: Applying wireless technology in real-time industrial process
  control,'' in \emph{2008 IEEE Real-Time and Embedded Technology and
  Applications Symposium}.\hskip 1em plus 0.5em minus 0.4em\relax IEEE, 2008,
  pp. 377--386.

\bibitem{gubbi2013internet}
J.~Gubbi, R.~Buyya, S.~Marusic, and M.~Palaniswami, ``Internet of things (iot):
  A vision, architectural elements, and future directions,'' \emph{Future
  generation computer systems}, vol.~29, no.~7, pp. 1645--1660, 2013.

\bibitem{houborkum09}
I.~Hou, V.~Borkar, and P.~R. Kumar, ``A theory of {Q}o{S} for wireless,'' in
  \emph{Proc. IEEE International Conference on Computer Communications
  (INFOCOM)}, Rio de Janeiro, Brazil, April 2009.

\bibitem{houkum09}
I.~Hou and P.~R. Kumar, ``Admission control and scheduling for {Q}o{S}
  guarantees for variable-bit-rate applications on wireless channels,'' in
  \emph{Proc. ACM international symposium on Mobile ad hoc networking and
  computing (MOBIHOC)}, New Orleans, Louisiana, May 2009.

\bibitem{kumar10}
------, ``Scheduling heterogeneous real-time traffic over fading wireless
  channels,'' in \emph{Proc. {IEEE} International Conference on Computer
  Communications ({INFOCOM})}, San Diego, California, March 2010.

\bibitem{srikant10}
J.~Jaramillo and R.~Srikant, ``Optimal scheduling for fair resource allocation
  in ad hoc networks with elastic and inelastic traffic,'' in \emph{Proc.
  {IEEE} International Conference on Computer Communications ({INFOCOM})}, San
  Diego, California, March 2010.

\bibitem{li2013optimal}
B.~Li and A.~Eryilmaz, ``Optimal distributed scheduling under time-varying
  conditions: A fast-csma algorithm with applications,'' \emph{IEEE
  Transactions on Wireless Communications}, vol.~12, no.~7, pp. 3278--3288,
  2013.

\bibitem{singh2018throughput}
R.~Singh and P.~Kumar, ``Throughput optimal decentralized scheduling of
  multihop networks with end-to-end deadline constraints: Unreliable links,''
  \emph{IEEE Transactions on Automatic Control}, vol.~64, no.~1, pp. 127--142,
  2018.

\bibitem{kang2014performance}
X.~Kang, W.~Wang, J.~J. Jaramillo, and L.~Ying, ``On the performance of
  largest-deficit-first for scheduling real-time traffic in wireless
  networks,'' \emph{IEEE/ACM Transactions on Networking}, vol.~24, no.~1, pp.
  72--84, 2014.

\bibitem{kang2015capacity}
X.~Kang, I.-H. Hou, L.~Ying \emph{et~al.}, ``On the capacity requirement of
  largest-deficit-first for scheduling real-time traffic in wireless
  networks,'' in \emph{Proceedings of the 16th ACM International Symposium on
  Mobile Ad Hoc Networking and Computing}.\hskip 1em plus 0.5em minus
  0.4em\relax ACM, 2015, pp. 217--226.

\bibitem{tsanikidis2020power}
C.~Tsanikidis and J.~Ghaderi, ``On the power of randomization for scheduling
  real-time traffic in wireless networks,'' \emph{arXiv preprint
  arXiv:2001.05146}, 2020.

\bibitem{hou2013scheduling}
I.-H. Hou, ``Scheduling heterogeneous real-time traffic over fading wireless
  channels,'' \emph{IEEE/ACM Transactions on Networking}, vol.~22, no.~5, pp.
  1631--1644, 2013.

\bibitem{dynkin2012theory}
E.~B. Dynkin, \emph{Theory of Markov processes}.\hskip 1em plus 0.5em minus
  0.4em\relax Courier Corporation, 2012.

\bibitem{neely2010queue}
M.~J. Neely, ``Queue stability and probability 1 convergence via lyapunov
  optimization,'' \emph{arXiv preprint arXiv:1008.3519}, 2010.

\bibitem{joo2015distributed}
C.~Joo, X.~Lin, J.~Ryu, and N.~B. Shroff, ``Distributed greedy approximation to
  maximum weighted independent set for scheduling with fading channels,''
  \emph{IEEE/ACM Transactions on Networking}, vol.~24, no.~3, pp. 1476--1488,
  2015.

\bibitem{robertson1996new}
N.~Robertson, D.~Sanders, P.~Seymour, and R.~Thomas, ``A new proof of the
  four-colour theorem,'' \emph{Electronic Research Announcements of the
  American Mathematical Society}, vol.~2, no.~1, pp. 17--25, 1996.

\bibitem{ross2013applied}
S.~M. Ross, \emph{Applied probability models with optimization
  applications}.\hskip 1em plus 0.5em minus 0.4em\relax Courier Corporation,
  2013.

\bibitem{chin2006online}
F.~Y. Chin, M.~Chrobak, S.~P. Fung, W.~Jawor, J.~Sgall, and T.~Tich{\`y},
  ``Online competitive algorithms for maximizing weighted throughput of unit
  jobs,'' \emph{Journal of Discrete Algorithms}, vol.~4, no.~2, pp. 255--276,
  2006.

\bibitem{jez2011one}
{\L}.~Je{\.z}, ``One to rule them all: A general randomized algorithm for
  buffer management with bounded delay,'' in \emph{European Symposium on
  Algorithms}.\hskip 1em plus 0.5em minus 0.4em\relax Springer, 2011, pp.
  239--250.

\bibitem{bienkowski2011randomized}
M.~Bienkowski, M.~Chrobak, and {\L}.~Je{\.z}, ``Randomized competitive
  algorithms for online buffer management in the adaptive adversary model,''
  \emph{Theoretical Computer Science}, vol. 412, no.~39, pp. 5121--5131, 2011.

\bibitem{jez2012online}
{\L}.~Je{\.z}, F.~Li, J.~Sethuraman, and C.~Stein, ``Online scheduling of
  packets with agreeable deadlines,'' \emph{ACM Transactions on Algorithms
  (TALG)}, vol.~9, no.~1, p.~5, 2012.

\end{thebibliography}

\end{document}